\documentclass[10pt]{amsart}

\usepackage{amssymb,amsmath,euscript}
\usepackage{mathtools}
\usepackage{breqn}
\usepackage{breqn}
\usepackage{tikz}
\usepackage{tikz}
\usepackage{tkz-berge}
\usepackage{enumitem}

\DeclarePairedDelimiter\floor{\lfloor}{\rfloor}

\usepackage{}

\newtheorem{theorem}{Theorem}[section]
\newtheorem{lemma}[theorem]{Lemma}
\newtheorem{proposition}[theorem]{Proposition}
\newtheorem{conjecture}[theorem]{Conjecture}
\newtheorem{corollary}[theorem]{Corollary}

\theoremstyle{definition}
\newtheorem{definition}[theorem]{Definition}

\theoremstyle{remark}
\newtheorem{remark}[theorem]{Remark}
\def\Fq{{\mathbb F}_q}
\def\a{{\alpha}}

\newcommand{\tr}{\operatorname{tr}}
\newcommand{\diag}{\operatorname{diag}}
\newcommand{\Ev}{\operatorname{Ev}}

\newcommand{\WH}{\operatorname{W_H}}
\newcommand{\rank}{\operatorname{rank}}

\newcommand{\RS}{\mathrm{Rowspan}}

\newcommand{\wt}{\mathrm{w_H}}

\newcommand{\supp}{\mathrm{Supp}}

\def\Psm{\mathbb{P}^{{m+1\choose 2}-1}}

\def\X{\mathbf{X}}
\def\fk{f^{\delta}_k(\X)}
\def\fkA{f^{\delta}_k(A)}

\def \Wktd{\mathrm{W}_k^{\delta}(t, m)}
\def \wkrd{\mathrm{w}_k^{\delta}( r, m)}
\def\Tkrd{T_k^\delta(r, m)}
\def\LBrm {\lambda_B(r, m)}
\def\LBr1m {\lambda_B(r-1, m)}

\def\GBarm {\gamma_{\a}(B, r, m)}

\def\Sm{S_m}

\def\PStm{\mathbf{\hat{\mathbb{S}}}(t, m)}

\def\Stm{\mathbb{S}(t, m)}

\def \C {C_{symm}(t, m)}
\def \CA {\widehat{C}_{symm}(t, m)}
\def \CAt {C_{symm}(2t, m)}

\begin{document}
	
	\title[Symmetric Determinantal Codes ]{Linear Codes Associated to Symmetric Determinantal Varieties: Even Rank Case}
	\author{Peter Beelen }
\address{Department of Applied Mathematics and Computer Science,\newline \indent
	Technical University of Denmark,\newline \indent
	Matematiktorvet 303B, 2800 Kgs. Lyngby, Denmark.}
\email{pabe@dtu.dk }

	\author{Trygve Johnsen}
\address{Department Mathematics and Statistics,\newline \indent
	UiT: The Arctic University of Norway,\newline \indent
	Hansine Hansens veg 18, 9019 Tromsø, Norway.}
\email{trygve.johnsen@uit.no}

	\author{Prasant Singh}
\address{Department of  Mathematics,\newline \indent
	IIT Jammu,\newline \indent
	 Jammu \& Kashmir, 181221.}
\email{psinghprasant@gmail.com}	

	\date{\today}
	
\begin{abstract}
	We consider linear codes over a finite field $\Fq$, for odd $q$, derived from  determinantal varieties, obtained from symmetric matrices of bounded ranks. A formula for the weight of a codeword is derived. Using this formula, we have computed the minimum distance for the codes corresponding to matrices upperbounded by any fixed, even rank. A conjecture is proposed for the cases where the upper bound is odd. At the end of the article, tables for the weights of these codes, for spaces of symmetric matrices up to order $5$, are given. 
	
	We also correct typographical errors in Proposition 1.1/3.1 of \cite{BJS}, and in the last table, and we have  rewritten Corollary 4.9 of that paper, and the usage of that Corollary in the proof of Proposition 4.10.

\end{abstract}

\maketitle

\section{Introduction}
One of the  most natural and interesting ways to construct classes of linear error-correcting codes is by using the language of projective systems and considering the $\Fq$-rational points of a projective algebraic variety defined over $\Fq$ as a projective system. For  a more detailed study of projective systems and linear codes associated to them, we refer to \cite[Chapter 1]{TVN}. Determinantal varieties in different projective spaces are  classical examples of projective algebraic varieties defined over finite fields, and therefore it is natural to study codes associated to these varieties. The study of codes associated to generic determinantal varieties was initiated by Beelen-Ghorpade-Hasan \cite{BGH}, and the minimum distance of the code associated to the variety defined by the vanishing of $2\times 2$ minors  of a generic matrix was determined. Later Beelen-Ghorpade \cite{BG} studied other codes associated to  generic determinantal varieties and together with several interesting properties, the minimum distances also of these codes were computed.

Determinantal varieties in the space of symmetric (skew-symmetric) matrices, i.e. the variety defined by the vanishing of all $(t+1)\times (t+1)$ minors of a generic symmetric (skew-symmetric) matrix of order $m$, where $t\le m$, is another example of classical varieties defined over a finite field $\Fq$. For a detailed exposition of these varieties, we refer to \cite{HTu, JLP}. Like generic determinantal varieties, these varieties are also an interesting class of $\Fq$-varieties with many $\Fq$-rational points. Therefore, it is natural to study the codes corresponding to these varieties. The study of codes associated to determinantal varieties in the space of skew-symmetric matrices was initiated by the first and the last named authors \cite{BS}. They found an iterative formula for the weight of a general codeword and using this formula they computed the minimum distance of these codes. In the present article, we consider the codes associated to determinantal varieties in the space of $\it{symmetric}$ matrices, and in the even rank cases, we compute the minimum distance of these codes.

Linear codes corresponding to symmetric determinantal varieties of rank at most $t$ can be obtained by evaluating homogeneous, linear functions in $m^2$ variables $X_{1,1},\ldots,X_{1,m},\ldots,X_{m,1},\ldots,X_{m,m}$ at points corresponding to symmetric matrices $[X_{i,j}]$ of rank at most $t$.
All important information about these codes  $\C$ can be read off from their shorter ``sister codes" $\CA$, each of which is obtained from $\C$, by skipping the zero matrices, and otherwise evaluating the linear functions in only one matrix in each multiplicative equivalence class of matrices of rank at most $t$.
In other words, we then consider the projective system of such matrices. We will consider these codes in parallel. Finding the length and dimensions of these codes is easy, and the following result will be proved in Section \ref{3}:
\begin{proposition} \label{first}
	
The linear code $\CA$ is an $[\hat{\nu}_m(t), \hat{k}]$ linear code where $\hat{\nu}_m(t)$ is given by equation \eqref{eq: carDet} below with
	\begin{equation}
\label{eq:LenDim}	
 {\nu}_m(t)= \sum\limits_{r=0}^t\left(\prod_{i=1}^{ \floor*{\frac{r}{2}}}\frac{q^{2i}}{q^{2i}-1}\prod_{i=0}^{r-1} (q^{m-i}-1)\right)\text{ and } \hat{k}={m+1\choose 2}.
\end{equation}
\end{proposition}
  To find the minimum distance and possible weights of codewords, however, a lot of calculations are needed. A first step in order to obtain control over the situation is taken in Propositions \ref{Prop: Diag} and \ref{squares}, where one shows that there are at most $2m+1$ weights that are possible.
In one of our main results, Theorem \ref{thm: weightoffk}, we determine these weights. Our results are analogues of
\cite[Corollary 3.3, Theorem 3.4]{BS}, and also of
\cite[Theorem 3]{N} for Grassmann codes $C(2,m)$. Theorem \ref{thm: weightoffk}, however, is not sufficiently transparent to read off the minimum distance of our codes immediately.  Let the symmetric affine determinantal variety $\mathbf{{\mathbb{S}}}(2t, m)$ be defined by the vanishing of all $(2t+1)\times(2t+1)$ minors of a generic symmetric matrix of size $m$.

The main result of this article is the following theorem:
\begin{theorem} \label{main}
   The minimum distance of the code $\widehat{C}_{symm}(2t, m)$  is
   $$
	 q^{m-1}\nu_{m-1}(2t-2) + q^{2t-1}\mu_{m-1}(2t-1)
	 + (q^{2t-1}-q^{-1})\mu_{m-1}(2t),
	 $$
  where $\nu_{m-1}(2t-2) =\lvert\mathbf{{\mathbb{S}}}(2t-2, m-1)\rvert$, and $\mu_{m-1}(2t-1)$ denotes the number of symmetric matrices of size $m-1$ and rank $2t-1$.
\end{theorem}
Theorem \ref{main}, in combination with Proposition \ref{first}, give the 3 basic parameters (length, dimension, minimum distance)  of our codes.

Theorem \ref{main} is parallel to \cite[Theorem 3.11]{BS} for codes from skew-symmetric matrices. The computational complexity, however, building on  Theorem \ref{thm: weightoffk},  is much larger than in the analogous computations in the skew-symmetric case, and to a great part, it involves a detailed treatment of quadratic hypersurfaces in projective and affine spaces and the three classes (hyperbolic, elliptic, and parabolic) of such hypersurfaces. In contrast to skew-symmetric matrices, the symmetric ones  also may have odd rank, but for reasons of space, we leave the determination of the minimum distances of $\C$ and $\CA$ for odd $t$ to future research. We end this paper, however, by giving a conjecture, and also a list of all  possible concrete weights appearing, for $m=3,4,5 $ and all $t$ up to $5$.
These results clearly indicate that our conjecture holds.

\section{Symmetric Determinantal Varieties over the field $\Fq$}
Let $m$ be a positive integer, $q$ be a prime power, and let $\Fq$ be the finite field with $q$ elements. An $m\times m$ matrix $A=(a_{ij})$ over $\Fq$ is called \emph{symmetric} if $a_{ij}= a_{ji}$ for every $1\leq i,\, j\leq m $. We use the notations $ \Sm$ to denote the space of $m\times m$ symmetric matrices over $\Fq$ and $A^T$ to denote the transpose of a matrix $A$. It is well known that $\Sm$ can be identified with the affine space $\mathbb{A}^{m+1\choose 2}$. Let $\X=(X_{ij})$ be an $m\times m$ matrix of independent indeterminates $X_{ij}$ over $\Fq$. For $0\leq t \leq m$, we denote by $\Stm$, the affine algebraic variety in $\Sm$ given by the vanishing of $(t+1)\times (t+1)$ minors of $\X$; in other words
\begin{equation}
\label{eq: SymmDet}
\Stm= \{A\in\Sm: \rank(A)\leq t\}
\end{equation}
For any $0\leq r\leq m$ we define
\[
S(r, m) = \{A\in\Sm: \rank(A)=r\}\qquad \text{ and }\qquad 
\mu_m(t)= |S(r,m)|.
\]
The following formula from \cite[Theorem 2]{MacWilliams} gives the number of symmetric matrices of size $m$ and rank $r$
\begin{equation}
\label{eq: SymmFix}
\mu_m(r)= \prod_{i=1}^{ \floor*{\frac{r}{2}}}\frac{q^{2i}}{q^{2i}-1}\prod_{i=0}^{r-1} (q^{m-i}-1).
\end{equation}
From the definition it is clear that
\[
\Stm = \bigcup\limits_{r=0}^t S(r, m)
\]
and the union is disjoint. Let $\nu_m(t)$
be the cardinality of the set $\Stm$. We have
\begin{equation}
\label{eq:FixRanDec}
\nu_m(t)=\sum\limits_{r=0}^t \mu_m(r).
\end{equation}
Note that the defining equations of $\Stm$ are square-free homogeneous equations of degree $t+1$. Consider the projective variety in $\Psm$ defined by these equations. This variety in the projective space $\Psm$ is known as the \emph{symmetric determinantal variety} and we denote it by $\PStm$. Note that $\Stm$ is the affine cone over $\PStm$. Therefore, if $\hat{\nu}_m(t)$ denote the cardinality of the set $\PStm$, then
\begin{equation}
\label{eq: carDet}
\hat{\nu}_m(t) = \frac{\nu_m(t)-1}{q-1}.
\end{equation}
\section{The Linear Code Associated to the Symmetric Determinantal Variety $\PStm$} \label{3}
We open this section by recalling the notion of a projective system. For details, we refer to \cite[Ch. 1]{TVN}. An $[n, k, d]$ projective system is a (multi)set $\Omega\subseteq\mathbb{P}^{k-1}$ that is not contained in a hyperplane of $\mathbb{P}^{k-1}$, and the parameters $n$ and $d$ are defined by
\begin{equation*}
|\Omega|=n,\qquad d=n- \max\{|\Omega\cap H|:\; H\text{ a hyperplane of }\mathbb{P}^{k-1}\}.
\end{equation*}
Every $[n, k, d]$ projective system corresponds to an equivalence class of $[n, k, d]$ linear code and vice-versa. Linear codes corresponding to several classical varieties have been studied by several authors. In this section, we are going to study the projective systems, and hence linear codes, associated to the determinantal varieties $\PStm$. First, we briefly give the construction of this linear code corresponding to $\PStm$.

Let $\Fq[\X]_1$ denote the vector space over $\Fq$ of linear homogeneous polynomials in $m^2$  variables $\X_{ij}$, for $1\leq i, j\leq m$. Let $\{P_1,\ldots, P_{\hat{\nu}_m(t)}\}\subseteq\Stm$ be a subset of representatives of distinct points of $\PStm$ in some fixed order. Consider the evaluation map

\begin{equation}
\label{eq:SymmCode}
\Ev: \Fq[\X]_1\to \mathbb{F}_q^{\hat{\nu}_m(t)}\; \text{defined by }f(\X)\mapsto c_f=(f(P_1),\ldots, f(P_{\hat{\nu}_m(t)})).
\end{equation}
The image of the evaluation map is a linear subspace of $ \mathbb{F}_q^{\hat{\nu}_m(t)}$. This subspace is called the  \emph{symmetric determinantal code} and is denoted by $\CA$.

\begin{proposition}
	\label{Prop:LenDim}
	The linear code $\CA$ is a $[\hat{\nu}_m(t), \hat{k}]$ linear code where $\hat{\nu}_m(t)$ is given by equation \eqref{eq: carDet} with
	\begin{equation}
	\label{eq:LenDimSymm}	
 {\nu}_m(t)= \sum\limits_{r=0}^t\left(\prod_{i=1}^{ \floor*{\frac{r}{2}}}\frac{q^{2i}}{q^{2i}-1}\prod_{i=0}^{r-1} (q^{m-i}-1)\right)\text{ and } \hat{k}={m+1\choose 2}.
\end{equation}
\end{proposition}
\begin{proof}
The proof of the proposition is simple. The statement about the length of the code $\CA$ follows from equations \eqref{eq: SymmFix},\;\eqref{eq:FixRanDec} and \eqref{eq: carDet}.
This gives the length of the code. For the dimension of the code, it is enough to prove that $\PStm$ does not lie on a hyperplane of $\Psm$. Equivalently, it is enough to show that for $1\leq t\leq m$, the span of the affine algebraic variety $\Stm$ is the vector space $\Sm$. Let $E_{ij}$ be the $m\times m$ matrix with $(i,j)^{\it th}$ entry $1$ and all other entries zero. It is well known that the set $\{E_{ij} + E_{ji}:1\leq i, j\leq m\} \cup \{E_{ii}: 1 \le i\le m\}$ is a basis of $\Sm$, therefore it is enough to show that the span of  $\Stm$ contains the set $\{E_{ij} + E_{ji}:1\leq i, j\leq m\}\cup \{E_{ii}: 1 \le i\le m\}$. Note that $S(1, m)\subseteq\Stm$ for every $1\leq t\leq m$. Therefore, we will show that for every $i, j$ the matrix $E_{ij}+E_{ji}$ can be written as a sum of rank one symmetric matrices. If $i=j$ then there is nothing to do as $E_{ii}$ itself is a rank one symmetric matrix. So assume $i<j$. Let $[X_{ij}]$ be the $m\times 1$ column vector with the $i^{\it th}$ and $j^{\it th}$ entries being one, and the rest of the coordinates being zero. Then $X_{ij}X_{ij}^T$ is a rank one symmetric matrix and
\[
E_{ij}+ E_{ji}= X_{ij}X_{ij}^T - E_{ii} - E_{jj}.
\]
This completes the proof.

\end{proof}
\begin{remark}
	\label{rmk: affine}
	Consider the linear map $\Ev_1:\Fq[\X]_1\to \Fq^{\nu_m(t)}$ defined by $f(\X)\mapsto (f(A_1),\ldots,f(A_{\nu_m(t)})$ where $\{A_1, \ldots, A_{\nu_m(t)}\}=\Stm$ in some fixed order. The image of this linear map is a subspace of $ \Fq^{\nu_m(t)}$ and we denote this subspace by $\C$. The codes $\CA$ and $\C$ are quite similar in the sense that the length of these codes is related by equation \eqref{eq: carDet} and both are of the same dimension. Further, let $f(\X)\in\Fq[\X]_1$ be a linear polynomial and $\widehat{c}_f\in\CA$ and  $c_f\in\C$ are the corresponding codewords, then the Hamming weights of these codewords are related by
	\begin{equation}
	\label{eq: wtafftoDet}
	\WH(c_f) =(q-1) \WH(\widehat{c}_f).
	\end{equation}
	Therefore to understand the minimum distance of the code $\CA$ it is enough to understand the minimum distance of the code $\C$. Therefore, for the rest of the section, we will concentrate on the code $\C$.
\end{remark}

For the rest of the article, we assume that $q$ is odd
(in the case of even $q$ one encounters difficulties with some of the calculations that we apply to obtain our results, but it is conceivable that the cases of even and odd $q$ are very similar. We make no claims about this case). Let $f=\sum f_{ij}X_{ij}\in\Fq[\X]_1$ be a  function and $c_f\in \C$ be the corresponding codeword. Let $F=(f_{ij})$ be the coefficient matrix of $f$ and $A$ be a symmetric matrix. Then
\[
f(A)= \tr(FA).
\]
Since the characteristic of the field is not $2$ and $A$ is a symmetric matrix we can find a symmetric matrix $G$ such that $\tr(FA)= \tr(G A)$ for every symmetric matrix $A$. Indeed, we choose $G $ to be the matrix whose $(i,j)^{\it th} $ entry is $\frac{f_{ij}+f_{ji}}{2}$ for $1\leq i, j\leq m$. Therefore, for every codeword $c\in \C$ there is a symmetric matrix $F=(f_{ij})$ such that the codeword $c$ is given by
\[
c= (\tr(FA_1), \ldots, \tr(FA_{\nu_m(t)}))
\]
where $A_1, \ldots, A_{\nu_m(t)}$ are matrices in $\Stm$. If $F$ is a matrix and $f(\X)\in\Fq[\X]_1$ is a function whose coefficient matrix is $F$, then we use the notation $c_F$ to denote the codeword $c_f\in\C$. In the next proposition, we see that corresponding to every codeword $c\in\C$, there is a diagonal matrix such that the weights of $c$ and the codeword corresponding to the diagonal matrix are the same.

\begin{proposition}
		\label{Prop: Diag}
	For every codeword $c\in \C$ there exists a diagonal matrix $F=\diag (1, 1,\ldots, 1, \delta, 0,\ldots 0)$ with  $\delta$ a non-zero element of $\Fq$, and
	\[
	\WH(c) = \WH(c_F).
	\]
\end{proposition}
\begin{proof}
For any non-singular $m\times m$ matrix $L$ and any $0\leq r\leq t$, the map $P\mapsto L^TPL$ gives a permutation of the set $S(r, m)$ and hence of the set $\Stm$. For any codeword $c$ we can find a symmetric matrix $F$ such that $c= c_{F}$. Since $F$ is symmetric, we can find \cite[Lemma 6]{MacWilliams} a nonsingular matrix $L$ such that $LFL^T =G =\diag (1, 1,\ldots, 1, \delta, 0,\ldots 0)$ as in the theorem. Now we claim that $G$ is the required matrix. Note that
\[
\tr(GP)= \tr(LF L^TP)= \tr(F(L^TPL)).
\]
This completes the proof of the theorem.	
\end{proof}
In the next proposition, we see that the weight of a codeword $c_F\in \C$ corresponding to a diagonal form $F$ as in the above proposition depends only on whether $\delta$ is a square or a non-square.

\begin{proposition}
\label{squares}
Let $F_i=\mathrm{diag}(1,1,\ldots,1,\delta_i,0,\ldots,0)$ for $i = 1, 2$ be two diagonal matrices of the same rank $k$ and let $c_{F_i}\in \C$ be the corresponding codewords. If $\delta_2/\delta_1$ is a square, then $W_H(c_{F_1})=W_H(c_{F_2})$.
\end{proposition}
\begin{proof}
By assumption, we can find a nonzero $b \in \mathbb{F}_q$ satisfying $b^2\delta_1=\delta_2$. Now consider the matrix $P=\mathrm{diag}(1,1,\ldots,1,b,1,\ldots,1)$ with $b$ at the $k^{\it th}$ position. Then $PF_1R^T=F_2$ and the map $M \mapsto P^TMP$ is a bijection between the support of $c_{F_2}$ and $c_{F_1}$. Therefore they have the same weight.
\end{proof}

	


From the above proposition, it is clear that for every $1\leq k\leq m$ there can be at most two different weights for codewords in $\C$. In other words, the weight of codewords corresponding to functions   $\fk= X_{11} + \cdots +X_{k-1k-1}+\delta X_{kk}$  for $1\leq k\leq m$ will be all possible weights of the code $\C$ and hence of the code  $\CA$. Moreover, the minimum of the weights $\wt(\fk)$ gives the minimum distance of the code $\C$. Next, we will give an iterative formula for the weight of codewords corresponding to functions $\fk$. But before that, we fix some notations for the rest of the article. First, from now on, whenever we  refer to the weight of the function $\fk$, we always mean the weight of the codeword $c_{\fk}\in\C$ corresponding to the function $\fk$.  For any $0\le k\le m$ and $0\le t\le m$, we define
$$
\supp_t(\fk):=\{A\in\Stm: \fkA\neq 0\} \text{ and }\Wktd=\left|\supp(\fk)\right|.
$$
Clearly, $\Wktd$ denotes the weight of the codeword $c_{\fk}\in \C$.
Also, for any fixed $0\le t\le m$ and   $0\le r\le t$, we define
$$
\Tkrd=\{A\in S(r, m): \fkA\neq 0\},\qquad \wkrd=|\Tkrd|.
$$
Notice that

\begin{equation}
\label{eq: restwtsum}
\supp_t(\fk)=\cup_{r=1}^t\Tkrd \text{ and }\Wktd =\sum\limits_{r=1}^t\wkrd
\end{equation}
We will call $\wkrd$ the $r^{\it th}$ restricted weight of the function $\fk$ in $\C$.

Note that for $k=0$, the function $\fk$ is the zero function,  and hence $w_0^\delta(r, m)=0$ for any $r$, and as a consequence  $W_0^\delta (t, m)=0$ for any $t$. Next, for any $B\in S(r, m)$ and $\a\in \Fq^*$, we define
$$
\lambda_B(r, m)=\left|\{vB: v\in\Fq^m\text{ and } vBv^T=0\} \right|
$$
and
$$
\gamma_{\a}(B, r, m)=\left|\{vB: v\in\Fq^m\text{ and } vBv^T+ \a=0\} \right|.
$$
(These entities are well defined since $vB=uB$ implies that $vBv^T=uBu^T$.)
Now let $1\le k\le m$ and $1\le t\le m$ be fixed. In view of equation \eqref{eq: restwtsum}, to determine the weight $\Wktd$, it is enough to determine $\wkrd$ for each $1\le r\le t$. Therefore, in the next lemma, we give a formula for $\wkrd$ in terms of $\LBrm$ and $\GBarm$.

\begin{lemma}
	\label{lemma:restwt}
For any $0\le r\le m$ and $k\ge 1$, 	the $r^{\it th}$ restricted weight of the codeword $\fk$ in $\C$ is given by

\begin{align}
\label{eq: restwt}
\wkrd &= (q-1)(q^{m-1}-q^{r-2})\mu_{m-1}(r-2)
+(q-2)q^{r-1}\mu_{m-1}(r-1) \\
 +\sum_{\mathclap{\substack{B\in S(r-1, m-1) \\ f_{k-1}^\delta(B)=0}}}\lambda_B(r-1, m-1)
& + \sum_{\mathclap{\a \in\Fq^*}}\sum_{\substack{B\in S(r-1, m-1) \nonumber\\ f_{k-1}^\delta(B)=\a}}\gamma_{\a}(B,r-1, m-1)- \sum_{\mathclap{\substack{B\in S(r, m-1) \\ f_{k-1}^\delta(B)=0}}}\lambda_B(r, m-1)  \nonumber\\
& -  \sum_{\mathclap{\a \in\Fq^*}}\sum_{\substack{B\in S(r, m-1) \\ f_{k-1}^\delta(B)=\a}}\gamma_{\a}(B,r, m-1) \nonumber
\end{align}

\end{lemma}
\begin{remark}
{\rm In the case $k=1$ this simplifies a lot, since
$f_{1-1}^{\delta}=f_0^{\delta}=0$, and moreover the weight of $f_1^{\delta}$ is equal to the weight of $f_1^{1}$                     for all $\delta$. }
\end{remark}

\begin{proof}
Consider the following map
\begin{align}
\label{eq: mappsi}
\Tkrd &\xrightarrow{\psi} S(r-2, m-1)\cup S(r-1, m-1)\cup S(r, m-1)\\
A &\mapsto  \psi(A)=B\nonumber
\end{align}
where $B$ is obtained from $A$ by deleting the first row and first column. Clearly, if $A\in S(r, m)$ then $\psi(A)=B$ is a symmetric matrix of size $m-1$ and of rank either $r-2,\; r-1$ or $r$. Therefore, the map $\psi$ is well-defined. Further, it is clear from the definition that $\wkrd =\left|\Tkrd\right|$, we use the map $\psi$ and fibers of this map to compute $\wkrd$. The main idea is to compute the cardinality  $\left| \psi^{-1}(B)\right|$ of the fibers of $\psi$ for every $B\in  S(r-2, m-1)\cup S(r-1, m-1)\cup S(r, m-1)$. We divide the counting of  $\left| \psi^{-1}(B)\right|$ into three different cases.

\noindent\textit{Case (I):} Suppose $B\in S(r-2, m-1)$. In this case,
$$
 \psi^{-1}(B)=\left\{ \begin{bmatrix}
 z & u \\
 u^T & B
 \end{bmatrix}: u\in \Fq^{m-1}\setminus \RS(B),\text{ and }z\neq -f_{k-1}^\delta(B)\right\}.
$$
Therefore, in this case, we get
\begin{equation}
\label{eq:typeI}
\left| \psi^{-1}(B)\right|= (q-1)(q^{m-1}- q^{r-2}).
\end{equation}

\noindent\textit{Case (II):} Suppose $B\in S(r-1, m-1)$. In this case, we have 	
$$
\psi^{-1}(B)= \left\{  \begin{bmatrix}
z &uB \\
Bu^T & B
\end{bmatrix}: u\in \Fq^{m-1},\; z\neq -f_{k-1}^\delta(B), \text{ and } z\neq uBu^T\right\}.
$$
The fibers in this case are divided into two classes.

\textit{Subcase (I):} Suppose $B\in S(r-1, m-1)$ and $f_{k-1}^\delta(B)=0$,  then

\begin{align*}
\psi^{-1}(B)=& \left\{  \begin{bmatrix}
z & uB \\
Bu^T & B
\end{bmatrix}: u\in \Fq^{m-1},\; uBu^T= 0, \text{ and }z\in \Fq^*  \right\}\\
\cup & \left\{  \begin{bmatrix}
z & uB \\
Bu^T & B
\end{bmatrix}: u\in \Fq^{m-1},\; uBu^T\neq 0, \text{ and }z\in \Fq^* \setminus\{uBu^T\} \right\},
\end{align*}
where $\Fq^*$ is the set set of nonzero elements of $\Fq$. Therefore, in this case, we get
\begin{align}
\label{eq:typeII1}
\left| \psi^{-1}(B)\right|&= (q-1)\lambda_B(r-1, m-1) + (q-2) (q^{r-1}-\lambda_B(r-1, m-1))\nonumber\\
&= (q-2)q^{r-1} + \lambda_B(r-1, m-1).
\end{align}

\textit{Subcase (II):} Suppose $B\in S(r-1, m-1)$ and $f_{k-1}^\delta(B)=\a$, where $\a\in \mathbb{F}_q^*$  then

\begin{align*}
\psi^{-1}(B)=& \left\{  \begin{bmatrix}
z & uB \\
Bu^T & B
\end{bmatrix}: u\in \Fq^{m-1},\; uBu^T+ \a= 0, \text{ and }z\neq -\a  \right\}\\
\cup & \left\{  \begin{bmatrix}
z & uB \\
Bu^T & B
\end{bmatrix}: u\in \Fq^{m-1},\; uBu^T+\a\neq 0, \text{ and }z\in \Fq \setminus\{-\a, uBu^T\} \right\}.
\end{align*}
Therefore, in this case, we get
\begin{align}
\label{eq:typeII2}
\left| \psi^{-1}(B)\right|&= (q-1)\gamma_\a(B, r-1, m-1) + (q-2) (q^{r-1}-\gamma_\a(B, r-1, m-1))\nonumber\\
&= (q-2)q^{r-1} + \gamma_\a(B, r-1, m-1).
\end{align}

\noindent\textit{Case (III):} Suppose $B\in S(r, m-1)$. In this case, we have 	
$$
\psi^{-1}(B)= \left\{  \begin{bmatrix}
uBu^T & uB \\
Bu^T & B
\end{bmatrix}: u\in \Fq^{m-1},\;\; uBu^T\neq -f_{k-1}^\delta(B) \right\}.
$$
The fibers in this case too are divided into two classes.

\textit{Subcase (I):} Suppose $B\in S(r, m-1)$ and $f_{k-1}^\delta(B)=0$,  then

\begin{align*}
\psi^{-1}(B)=& \left\{  \begin{bmatrix}
uBu^T & uB \\
Bu^T & B
\end{bmatrix}: u\in \Fq^{m-1},\;\; uBu^T\neq 0 \right\}.
\end{align*}
Therefore, in this case, we get
\begin{equation}
\label{eq:typeIII1}
\left| \psi^{-1}(B)\right|= (q^r- \lambda_B(r, m-1)).
\end{equation}

\textit{Subcase (II):} Suppose $B\in S(r, m-1)$ and $f_{k-1}^\delta(B)=\a$, where $\a\in \mathbb{F}_{q}^*$  then

\begin{align*}
\psi^{-1}(B)=& \left\{  \begin{bmatrix}
uBu^T & uB \\
Bu^T & B
\end{bmatrix}: u\in F_{q^{m-1}},\;\; uBu^T +\a\neq 0  \right\}.
\end{align*}
Therefore, in this case, we get
\begin{equation}
\label{eq:typeIII2}
\left| \psi^{-1}(B)\right|= (q^r- \gamma_\a(B, r, m-1)).
\end{equation}
Now, taking the map $\psi$ and its fibers into account, we get

\begin{align*}
\wkrd &= \sum_{\mathclap{\substack{B\in S(r-2, m-1) }}}\left| \psi^{-1}(B)\right| \quad + \quad\sum_{\mathclap{\substack{B\in S(r-1, m-1) }}} \left| \psi^{-1}(B)\right| \quad + \quad \sum_{\mathclap{\substack{B\in S(r, m-1) }}} \left| \psi^{-1}(B)\right|\\
&= \sum_{\mathclap{\substack{B\in S(r-2, m-1) }}}\left| \psi^{-1}(B)\right| \quad + \quad \sum_{\substack{B\in S(r-1, m-1) \\ f_{k-1}^\delta(B)=0}} \left| \psi^{-1}(B)\right| \\
&+\quad \sum\limits_{\a\in\Fq^*}  \sum_{\substack{B\in S(r-1, m-1) \\ f_{k-1}^\delta(B)=\a}} \left| \psi^{-1}(B)\right| \quad + \quad \sum_{\substack{B\in S(r, m-1) \\ f_{k-1}^\delta(B)=0}} \left| \psi^{-1}(B)\right| \\
&+\quad \sum\limits_{\a\in\Fq^*}  \sum_{\substack{B\in S(r, m-1) \\ f_{k-1}^\delta(B)=\a}} \left| \psi^{-1}(B)\right|.
\end{align*}
Using equations \eqref{eq:typeI}, \eqref{eq:typeII1}, \eqref{eq:typeII2}, \eqref{eq:typeIII1} and \eqref{eq:typeIII2}, we get

\begin{align*}
\wkrd &= (q-1)(q^{m-1}-q^{r-2}) \mu_{m-1}(r-2) + \sum_{\substack{B\in S(r-1, m-1) \\ f_{k-1}^\delta(B)=0}} ((q-2)q^{r-1} + \lambda_B(r-1, m-1))\\
&+ \sum\limits_{\a\in\Fq^*}  \sum_{\substack{B\in S(r-1, m-1) \\ f_{k-1}^\delta(B)=\a}} ((q-2)q^{r-1} + \gamma_\a(B, r-1, m-1))+\; \sum_{\substack{B\in S(r, m-1) \\ f_{k-1}^\delta(B)=0}} (q^r- \lambda_B(r, m-1)) \\
&+ \sum\limits_{\a\in\Fq^*} \;\; \sum_{\substack{B\in S(r, m-1) \\ f_{k-1}^\delta(B)=\a}} (q^r- \gamma_\a(B, r, m-1)).\\
&= (q-1)(q^{m-1}-q^{r-2})\mu_{m-1}(r-2) +\sum_{\substack{B\in S(r-1, m-1) }} (q-2)q^{r-1} \\
& + \sum_{\substack{B\in S(r-1, m-1) \\ f_{k-1}^\delta(B)=0}} \lambda_B(r-1, m-1) +  \sum\limits_{\a\in\Fq^*}\;\;\sum_{\substack{B\in S(r-1, m-1) \\ f_{k-1}^\delta(B)=\a}} \gamma_\a(B, r-1, m-1)\\
&+\sum_{\substack{B\in S(r, m-1) }} q^r -\sum_{\substack{B\in S(r, m-1) \\ f_{k-1}^\delta(B)=0}} \lambda_B(r, m-1)) \\
&-  \sum\limits_{\a\in\Fq^*} \;\; \sum_{\substack{B\in S(r, m-1) \\ f_{k-1}^\delta(B)=\a}}  \gamma_\a(B, r, m-1).
\end{align*}
After simplifying this, we get

\begin{align*}
\wkrd &= (q-1)(q^{m-1}-q^{r-2})\mu_{m-1}(r-2) + q^{r}\mu_{m-1}(r-1) \\
& -2q^{r-1}\mu_{m-1}(r-1)  + q^r\mu_{m-1}(r) + \sum_{\mathclap{\substack{B\in S(r-1, m-1) \\ f_{k-1}^\delta(B)=0}}}\lambda_B(r-1, m-1) \nonumber\\
& + \sum_{\mathclap{\a \in\Fq^*}}\;\sum_{\substack{B\in S(r-1, m-1) \nonumber\\ f_{k-1}^\delta(B)=\a}}\gamma_{\a}(B,r-1, m-1)- \sum_{\mathclap{\substack{B\in S(r, m-1) \\ f_{k-1}^\delta(B)=0}}}\lambda_B(r, m-1)  \nonumber\\
& -  \sum_{\mathclap{\a \in\Fq^*}}\;\sum_{\substack{B\in S(r, m-1) \\ f_{k-1}^\delta(B)=\a}}\gamma_{\a}(B,r, m-1).
\end{align*}
This completes the proof of the lemma.
\end{proof}	
 Now we are ready to prove one of the main theorems of this article.
 \begin{theorem}
 	\label{thm: weightoffk}
 	For any $0\le k\le m$ and $0\le t\le m$, the Hamming weight of the codeword $\fk\in\C$ is given by
 	\begin{align*}
 		\Wktd &= \theta_q(t, m) - \sum_{\mathclap{\substack{B\in S(t, m-1) \\ f_{k-1}^\delta(B)=0}}}\lambda_B(t, m-1) - \sum_{\mathclap{\a \in\Fq^*}}\;\;\sum_{\substack{B\in S(t, m-1) \\ f_{k-1}^\delta(B)=\a}}\gamma_{\a}(B,t, m-1)
 	\end{align*}
  where
  \begin{equation}
      \label{eq: theta}
      \theta_q(t, m)=(q-1)q^{m-1}\nu_{m-1}(t-2) + (q-1)q^{t-1}\mu_{m-1}(t-1) + q^t\mu_{m-1}(t)
  \end{equation}
 \end{theorem}
\begin{proof}
The proof is a simple consequence of equation 	\eqref{eq: restwtsum} and Lemma \ref{lemma:restwt}. We have
\begin{align*}
\Wktd &=\sum\limits_{r=1}^t\wkrd\\
&= \sum_{r=1}^t (q-1)q^{m-1}\mu_{m-1}(r-2)-\sum_{r=1}^t (q-1)q^{r-2}\mu_{m-1}(r-2)\\
&+\sum_{r=1}^t (q-2)q^{r-1}\mu_{m-1}(r-1) +\sum_{r=1}^t q^r \mu_{m-1}(r)\\
&+ \sum_{r=1}^t\qquad \sum_{\mathclap{\substack{B\in S(r-1, m-1) \\ f_{k-1}^\delta(B)=0}}}\lambda_B(r-1, m-1) +\sum_{r=1}^t \sum_{\mathclap{\a \in\Fq^*}}\;\sum_{\substack{B\in S(r-1, m-1) \nonumber\\ f_{k-1}^\delta(B)=\a}}\gamma_{\a}(B,r-1, m-1)\\
&- \sum_{r=1}^t\qquad\sum_{\mathclap{\substack{B\in S(r, m-1) \\ f_{k-1}^\delta(B)=0}}}\lambda_B(r, m-1) - \sum_{r=1}^t \sum_{\mathclap{\a \in\Fq^*}}\;\sum_{\substack{B\in S(r, m-1) \\ f_{k-1}^\delta(B)=\a}}\gamma_{\a}(B,r, m-1).
\end{align*}
Note that $\mu_{m-1}(0)=1$, since $S(0,m-1)$ only contains the zero matrice. This also implies that $\sum_{B \in S(0,m-1)} \lambda_B(0,m-1)=\lambda_0(0,m-1)=1$, and $\gamma_\alpha(B,0,m-1)=0$ for all non-zero $\alpha$. Now expanding the above expression in $r$, and after canceling several terms, we get

\begin{align*}
 		\Wktd &= \theta_q(t, m) - \sum_{\mathclap{\substack{B\in S(t, m-1) \\ f_{k-1}^\delta(B)=0}}}\lambda_B(t, m-1) - \sum_{\mathclap{\a \in\Fq^*}}\;\;\sum_{\substack{B\in S(t, m-1) \\ f_{k-1}^\delta(B)=\a}}\gamma_{\a}(B,t, m-1)
 	\end{align*}

This completes the proof.	
	
\end{proof}	
\begin{remark}
For $t=m$, Theorem \ref{thm: weightoffk} gives the obvious formula $f_k^{\delta}=q^{{m +1 \choose\ 2}}-q^{{m+1 \choose\ 2}-1}$, for all $k$ and $\delta$.
\end{remark}

Every $A\in S(r, m)$ corresponds to a quadratic form  $Q_A(X)=XAX^T$. Since $A$ is of rank $r$, one can change $Q_A(X)$, using a linear change of variables, into one of the following forms \cite[Theorem 5.2.4]{Hir}.
\begin{enumerate}
	\item If $r= 2t+1$, then the quadratic $Q_A(X)$ can be transformed to
	$$
	cX_1^2 + X_2X_3+\cdots + X_{2t}X_{2t+1}.
	$$
	for some non-zero constant $c$ and independent linear forms $X_i$.
	\item  If $r=2t$, then the quadratic $Q_A(X)$ can be transformed to either
	$$
	(i)\;\; c( X_1X_2 +\cdots+ X_{2t-1}X_{2t}) \text{ or }\newline
	(ii)\; c((f(X_1, X_2)+ X_3X_4+\cdots+ X_{2t-1}X_{2t})
	$$
	for some non-zero constant $c$ and independent linear forms $X_i$, and
	where $f(X_1, X_2)$ is an irreducible, homogeneous polynomial of degree $2$ (it can even be brought to the form
	$X_1^2-aX_2^2$ for $a$ a non-square).
\end{enumerate}
\begin{definition} \label{standard}
The expressions in (1) and(2)(i),(ii)  will be denoted by \emph{standard forms} of parabolic, hyperbolic, and elliptic quadratic forms, respectively. We will in general denote a symmetric matrix $A$ as being of parabolic type, hyperbolic type, or elliptic type, depending on whether the quadratic form  $Q_A(X)$ can be transformed to a parabolic,  hyperbolic, or elliptic quadratic form of standard type, as described above, respectively.
\end{definition}
 From the discussion above, it is clear that if $r$ is odd, then every matrix $A\in S(r, m)$ is of parabolic type. On the other hand, if $r$ is even, then a matrix $A\in S(r, m)$ is either of hyperbolic or elliptic type. For every $1\le r\le m/2$, we define
$$
v_{+1}(2r, m)=\left|\{A\in S(2r, m): A\text{ is hyperbolic }\}\right|
$$
 and
$$
v_{-1}(2r, m)=\left|\{A\in S(2r, m): A\text{ is elliptic }\}\right|.
$$
From \cite[Prop. 2.4]{Schmidt}, we know that
\begin{equation}
\label{eq: noofhyp}
v_{\pm 1}(2r, m)=\frac{q^r\pm 1}{2}\dfrac{\prod\limits_{i=0}^{2r-1}(q^m-q^i)}{\prod\limits_{i=0}^{r-1}(q^{2r}-q ^{2i})}
\end{equation}

Recall that, for $B\in S(r, m)$ we denote by $\lambda_B(r, m)$ the number of $XB\in\Fq^m$ such that $XBX^T=0$. From the fact that a matrix $B\in S(r,m)$ has nullity $m-r$ and from \cite[Thm 5.2.6]{Hir}, we have that
\begin{eqnarray}
\label{eq: noofsoln}
\lambda_B(r, m)= \begin{cases*}
q^{2t}, \text{ if }r=2t+1,\\
q^{2t-1} + q^{t}- q^{t-1},\; \text{ if }r=2t \text{ and B is hyperbolic},\\
q^{2t-1} - q^{t} + q^{t-1},\; \text{ if }r=2t \text{ and B is elliptic.}
\end{cases*}
\end{eqnarray}
Recall also that, for any $B\in S(r, m)$ and $\a\in\Fq^*$, we have defined $\gamma_{\a}(B, r, m)$ as the number of $uB\in\Fq^m$ such that $uBu^T +\a=0$. If $B_\a=\diag(B, \a)$, then it is not hard to see that
\begin{equation}
\label{eq: gammaalpha}
\gamma_{\a}(B, r, m)= \left|\{[u:v]B_\a\in\mathbb{P}^m:uBu^T +\a v^2=0\}\right|-\left|\{\overline{uB}\in\mathbb{P}^{m-1}:uBu^T =0\}\right|.
\end{equation}
By $\overline{uB}$ for non-zero $uB$ we mean its image in
$\mathbb{P}^{m-1}$ via multiplicative equivalence. Thus, $\gamma_{\a}(B, r,m)$ is the difference of the projective solution of the quadratic forms defined by matrix $B$, and matrix $B_\a$.
Note that, with more detailed information about $B$ and $\a$, then using equation \eqref{eq: noofsoln} we can compute the value of $\gamma_{\a}(B, r,m)$. For example, if the rank of $B$ is even, then $B_\a$ is parabolic for every $\a\in\Fq^*$  and therefore, we can find $\gamma_{\a}(B,r, m)$ if we know the nature of $B$. On the other hand, if the rank of $B$ is odd, then we can transform the quadratic form $XBX^T$ into $cX_1^2 + X_2X_3+\cdots + X_{2t}X_{2t+1}$. Now, if we know $\a\in\Fq^*$, we can know whether $X_1^2 +\a Y^2$ is reducible or irreducible (depending on whether
$-\frac{\a}{c}$ is a square or not, referring to formula
(1) above, and hence we know that $B_\a$ is hyperbolic or elliptic respectively. Therefore, we can find $\gamma_{\a}(B,r, m)$. Next, we are ready to compute the Hamming weight of the codeword $f_1^\delta(X)=\delta X_{1,1}$. Clearly, the weight of $f_1^\delta(X)=\delta X_{1,1}$ is independent of $\delta$.

\begin{theorem}
	\label{thm: WeightW1}
	The Hamming weight of the codeword $f_1^\delta(X)\in\C$ is given by
	\begin{eqnarray*}
	W_1^\delta(t, m)=\begin{cases*}
	(q-1)q^{m-1}\nu_{m-1}(t-2) + (q-1)q^{t-1}\mu_{m-1}(t-1)\\
	 + (q-1)q^{t-1}\mu_{m-1}(t),	\text{ if }t \text{ is odd },\\
	 \\
(q-1)q^{m-1}\nu_{m-1}(t-2) + (q-1)q^{t-1}\mu_{m-1}(t-1) \\
	 + (q-1)(q^{t-1}-q^{-1})\mu_{m-1}(t),	 	\text{ if }t \text{ is even }.\\
	\end{cases*}
	\end{eqnarray*}
\end{theorem}
\begin{proof}
	From Theorem \ref{thm: weightoffk}, we have
		\begin{align*}
	W_1^\delta(t, m) &= (q-1)q^{m-1}\nu_{m-1}(t-2) + (q-1)q^{t-1}\mu_{m-1}(t-1) + q^t\mu_{m-1}(t-1)\\\nonumber
	& - \sum_{\mathclap{\substack{B\in S(t, m-1) \\ f_{0}^\delta(B)=0}}}\lambda_B(t, m-1) - \sum_{\mathclap{\a \in\Fq^*}}\;\;\sum_{\substack{B\in S(t, m-1) \\ f_{0}^\delta(B)=\a}}\gamma_{\a}(B,t, m-1)
	\end{align*}
Note that, $f_0^\delta(X)$	is the zero function and therefore for every $1\le t\le m$ and $B\in S(t, m-1)$, we have $f_{0}^\delta(B)=0$. Also, in view of the above expression, it is clear that to determine $	W_1^\delta(t, m)$, it is enough to determine
$$
\Lambda(t, m-1)=\sum_{\mathclap{\substack{B\in S(t, m-1) \\ f_{0}^\delta(B)=0}}}\lambda_B(t, m-1).
$$
Note that the sum in the expression for $\Lambda(t, m-1)$ is over the set $S(t, m-1)$, since for any $B\in S(t, m-1)$, we have $ f_{0}^\delta(B)=0$. Now we divide the proof  into two cases.
\begin{enumerate}
	\item[Case 1.] If $t$ is odd: In this case, every matrix $B\in S(2t+1, m-1)$ is parabolic and therefore, from equation \eqref{eq: noofsoln}, we get
	$\lambda_B(t, m-1)= q^{t-1}$ for every $B\in S(t, m-1)$. Therefore,  we get
	$$
	\Lambda(t, m-1)=\mu_{m-1}(t) q^{t-1}.
	$$
\end{enumerate}
Putting this value in the expression for $	W_1^\delta(t, m)$, we have
	\begin{align*}
W_1^\delta(t, m) &= (q-1)q^{m-1}\nu_{m-1}(t-2) + (q-1)q^{t-1}\mu_{m-1}(t-1)\\
&  + q^t\mu_{m-1}(t)- \mu_{m-1}(t) q^{t-1}\\
&= (q-1)q^{m-1}\nu_{m-1}(t-2) + (q-1)q^{t-1}\mu_{m-1}(t-1)\\
& + (q-1)q^{t-1}\mu_{m-1}(t).
\end{align*}
\item[ Case 2.]If $t$ is even: In this case, there are $v_{+1}(t, m-1)$ hyperbolic matrices and $v_{-1}(t, m-1)$ elliptic matrices $B\in S(t, m-1)$. From \eqref{eq: noofsoln}, we know that if $B\in S(t, m-1)$ is hyperbolic, then $\lambda_B(t, m-1)= q^{t-1}+ q^{t/2}- q^{t/2 -1}$ and if if $B\in S(t, m-1)$ is elliptic, then $\lambda_B(t, m-1)= q^{t-1}- q^{t/2} + q^{t/2 -1}$. Therefore, we get
	\begin{align*}
W_1^\delta(t, m) &= (q-1)q^{m-1}\nu_{m-1}(t-2) + (q-1)q^{t-1}\mu_{m-1}(t-1) + q^t\mu_{m-1}(t)\\&
- v_{+1}(t, m-1)(q^{t-1}+ q^{t/2}- q^{t/2 -1}) -  v_{-1}(t, m-1)(q^{t-1}- q^{t/2}+ q^{t/2 -1})\\
&= (q-1)q^{m-1}\nu_{m-1}(t-2) + (q-1)q^{t-1}\mu_{m-1}(t-1) + q^t\mu_{m-1}(t)\\
& - q^{t-1} \mu_{m-1}(t) - (q^{t/2}- q^{t/2 -1})(v_{+1}(t, m-1) - v_{-1}(t, m-1))\\
&= (q-1)q^{m-1}\nu_{m-1}(t-2) + (q-1)q^{t-1}\mu_{m-1}(t-1)\\
& + (q-1)q^{t-1}\mu_{m-1}(t) -(q-1)q^{t/2 -1}\dfrac{\prod\limits_{i=0}^{t-1}(q^{m-1}-q^i)}{\prod\limits_{i=0}^{t/2-1}(q^{t}-q ^{2i})}
\end{align*}
where we are using the the facts that $v_{+1}(t, m-1) + v_{-1}(t, m-1)= s(t, m-1)$ and  $v_{+1}(t, m-1) - v_{-1}(t, m-1)= \frac{\prod\limits_{i=0}^{t-1}(q^{m-1}-q^i)}{\prod\limits_{i=0}^{t/2-1}(q^{t}-q ^{2i})}=q^{\frac{-t}{2}}s(t,m-1)$. This completes the proof of the theorem.
\end{proof}
%
%
%
%

\section{The minimum weight of the codes codes $\C$ and $\CA$ }
In this section, we determine the minimum distance of  codes $\C$ and hence of codes $\CA$ in the even rank cases. For convenience, we choose to denote the rank by $2t$, and not $t$, then. In order to prove the main theorem of the section, we need to introduce the following notations that are going to be useful. For any $0\le 2t\le m$, and $1\le k\le m$, we define sets
\begin{align*}
P_k^{\delta}(2t-1, m) &=\{B\in S(2t-1, m): f_k^\delta(B)=0\},\\
H_k^{\delta}(2t, m) & =\{B\in S(2t, m): B\text{ is hyperbolic and }f_k^\delta(B)=0\}, \text{ and }\\
E_k^{\delta}(2t, m) &=\{B\in S(2t, m):B\text{ is elliptic and } f_k^\delta(B)=0\}.
\end{align*}

Set
$$
p_k^{\delta}(2t-1, m)=|P_k^{\delta}(2t-1, m)|,\; h_k^{\delta}(2t, m)=|H_k^{\delta}(2t, m)|,\text{ and } e_k^{\delta}(2t, m)=|E_k^{\delta}(2t, m)|.
$$

Note that $h_0^{\delta}(2t, m)=v_{+1}(2t, m)$ and  $e_0^{\delta}(2t, m)=v_{-1}(2t, m)$. We want to prove that $W_1^\delta(2t, m)$ is the minimum distance of the code $\CAt$. The next theorem gives a sufficient condition for the minimality of the weight $W_1^\delta(2t, m)$.

\begin{theorem}
	\label{thm: wtdiff}
	For any $1\leq k\leq m$, the difference $W_k^\delta(2t, m)- W_1^\delta(2t, m)$ of the weights $W_k^\delta(2t, m)$ and $W_1^\delta(2t, m)$  is given by
	$$
	 q^t\left( ( v_{+1}(2t, m-1)- v_{-1}(2t, m-1))- (h_{k-1}^{\delta}(2t, m-1)- e_{k-1}^{\delta}(2t, m-1))\right).
	$$
\end{theorem}

\begin{proof}
Recall that, from Theorem \ref{thm: weightoffk}, we have
\begin{align*}
 		W_k^\delta(2t, m) &= \theta_q(t, m) - \sum_{\mathclap{\substack{B\in S(t, m-1) \\ f_{k-1}^\delta(B)=0}}}\lambda_B(t, m-1) - \sum_{\mathclap{\a \in\Fq^*}}\;\;\sum_{\substack{B\in S(t, m-1) \\ f_{k-1}^\delta(B)=\a}}\gamma_{\a}(B,t, m-1)
 	\end{align*}
  where $\theta_q(t, m)$ is given by equation \eqref{eq: theta}. Thus, $W_k^\delta(2t, m) - W_1^\delta(2t, m)$ is equal to
\begin{align*}
& \sum_{\mathclap{\substack{B\in S(2t, m-1) \\ f_{0}^\delta(B)=0}}}\lambda_B(2t, m-1) \quad+ \quad\sum_{\mathclap{\a \in\Fq^*}}\sum_{\substack{B\in S(2t, m-1) \\ f_{0}^\delta(B)=\a}}\gamma_{\a}(B,2t, m-1)\\
& - \sum_{\mathclap{\substack{B\in S(2t, m-1) \\ f_{k-1}^\delta(B)=0}}}\lambda_B(2t, m-1)\quad - \quad\sum_{\mathclap{\a \in\Fq^*}}\;\;\sum_{\substack{B\in S(2t, m-1) \\ f_{k-1}^\delta(B)=\a}}\gamma_{\a}(B,2t, m-1)\\
&=  \sum_{\mathclap{\substack{B\in S(2t, m-1) \\ B\text{ is hyperbolic}}}}\lambda_B(2t, m-1) \quad+\quad \sum_{\mathclap{\substack{B\in S(2t, m-1) \\ B\text{ is elliptic}}}}\lambda_B(2t, m-1) \quad- \quad\sum_{\mathclap{\substack{B\in S(2t, m-1) \\ B\text{ hyp.; } f_{k-1}^\delta(B)=0}}}\lambda_B(2t, m-1) \\
& - \quad\sum_{\mathclap{\substack{B\in S(2t, m-1) \\ B\text{ ellip.; } f_{k-1}^\delta(B)=0}}}\lambda_B(2t, m-1) \quad-\quad \sum_{\mathclap{\substack{B\in S(2t, m-1) \\ B\text{ hyp.; } f_{k-1}^\delta(B)\neq 0}}}\gamma_{f_{k-1}^\delta(B)}(B,2t, m-1)\\\nonumber
& - \quad \sum_{\mathclap{\substack{B\in S(2t, m-1) \\ B\text{ ellip.; } f_{k-1}^\delta(B)\neq 0}}}\gamma_{f_{k-1}^\delta(B)}(B,2t, m-1)\\
 &=  v_{+1}(2t, m-1) (q^{2t-1}+ q^t- q^{t-1}) + v_{-1}(2t, m-1) (q^{2t-1}- q^t+ q^{t-1})\\
 & - h_{k-1}^{\delta}(2t, m-1)(q^{2t-1}+ q^t- q^{t-1}) - e_{k-1}^{\delta}(2t, m-1)(q^{2t-1}- q^t+ q^{t-1})\\
 &- (v_{+1}(2t, m-1)- h_{k-1}^{\delta}(2t, m-1)) q^{t-1}(q^t-1)\\
  &- (v_{-1}(2t, m-1)- e_{k-1}^{\delta}(2t, m-1)) q^{t-1}(q^t+1).
\end{align*}
Simplifying the above expression, we get the desired result.
	
\end{proof}	
The next corollary then gives a necessary and sufficient condition for $W_1^\delta(2t, m) $ to be the minimum weight of the code $\CAt$.
\begin{corollary}
	\label{cor: wtdiff}
	The weight $W_1^\delta(2t, m) $ is the minimum weight of the code $\CAt$ iff
	$$
	v_{+1}(2t, m-1)- v_{-1}(2t, m-1)\geq h_{k-1}^{\delta}(2t, m-1)- e_{k-1}^{\delta}(2t, m-1).
	$$
\end{corollary}

\begin{proof}
The corollary follows from the last theorem as 	
\begin{align*}
&W_k^\delta(2t, m) - W_1^\delta(2t, m)\geq 0\\
&\iff  	v_{+1}(2t, m-1)- v_{-1}(2t, m-1)\geq h_{k-1}^{\delta}(2t, m-1)- e_{k-1}^{\delta}(2t, m-1).
\end{align*}
	
\end{proof}	

It is clear from Corollary \ref{cor: wtdiff}, that to prove that  $W_1^\delta(2t, m)$ is the minimum distance of  $\CAt$, we need to estimate $h_{k-1}^{\delta}(2t, m-1)- e_{k-1}^{\delta}(2t, m-1)$, for arbitrary $k$, $t$, and $m.$ We will do this by obtaining an  upper bound for $h_{k-1}^{\delta}(2t, m-1)$  and a lower bound for $e_{k-1}^{\delta}(2t, m-1)$. For practical, notational reasons we will replace $m-1$ by $m$ in the following, and find an upper bound for $h_{k}^{\delta}(2t, m)$ and a lower bound for $e_{k}^{\delta}(2t, m)$ and compare the difference $h_{k}^{\delta}(2t, m)- e_{k}^{\delta}(2t, m)$ with   $	v_{+1}(2t, m)- v_{-1}(2t, m)$.  To get these bounds, we consider the following restriction of the map $\psi$ defined in equation \eqref{eq: mappsi} to the sets
$H_k^{\delta}(2t, m)$ and  $E_k^{\delta}(2t, m)$ to obtain maps $\phi_1$ and $\phi_2$:
\begin{align}
\label{eq: restmapchiI}
H_{k}^\delta(2t, m) &\xrightarrow{\phi_1} S(2t-2, m-1)\cup S(2t-1, m-1)\cup S(2t, m-1)\\
A &\mapsto  \phi_1(A)=B\nonumber
\end{align}
 and
 \begin{align}
 \label{eq: restmapchiII}
 E_{k}^\delta(2t, m) &\xrightarrow{\phi_2} S(2t-2, m-1)\cup S(2t-1, m-1)\cup S(2t, m-1)\\
 A &\mapsto  \phi_2(A)=B\nonumber
 \end{align}
where $B=\phi_i(A)$ for $i=1, 2$ is obtained from $A$ by deleting the first row and first column. Note that $\phi_i$ is just the restriction of the map $\psi$ defined in equation \eqref{eq: mappsi} to appropriate subsets. But we use different notations only to make things simpler and to not worry much about the domain where the map $\psi$ is getting restricted.
The idea is to use these maps and their fibers to estimate $ h_k^\delta(2t, m)$ from above and $e_k^\delta(2t, m)$ from below.  We count the fibers of these maps in several lemmas. In the  next lemma, we give the structure of the fibers $\phi_1^{-1}(B)$ for $B\in S(2t-2, m-1)$ and determine $|\phi_1^{-1}(B)|$ for such $B$.

\begin{lemma}
	\label{lemma: fiberIhyp}
	Let $B\in S(2t-2, m-1)$ and $A\in S(2t, m)$ be such that by deleting the first row and first column of $A$ we get $B$. Then $A$ is hyperbolic iff $B$ is hyperbolic. Further, If $B\in S(2t-2, m-1)$, then
	$$
	|\phi_1^{-1}(B)|=\begin{cases*}
	q^{m-1}-q^{2t-2}, \text{ if B is hyperbolic,} \\
	0, \quad\quad\quad\quad\quad\text{ otherwise.}
	\end{cases*}
	$$
\end{lemma}

\begin{proof}
Let $B\in S(2t-2, m-1)$	and $L\in GL_{m-1}(\Fq)$ be a matrix such that $LBL^T=D$ is on the standard hyperbolic or elliptic form described in Definition \ref{standard}. Let
\[
A=\begin{bmatrix}
\a & u \\
u^T & B
\end{bmatrix}.
\]
If we take
$$
M_1=\begin{bmatrix}
1 & 0 \\
0^T & L
\end{bmatrix},
$$
then
$$
M_1AM_1^T=\begin{bmatrix}
\a & z \\
z^T & D
\end{bmatrix},
$$	
where $z=uL^T$. Note that $\rank A= \rank B +2$ if and only if $z\notin \RS(D)$. If $A\in S(2t, m)$, we may write $z= yD + b$, where $b= (0,\ldots, 0, b_{2t-1},\ldots, b_{m-1})$ and $b\neq 0$, i.e., $b_i\neq 0$ for some $i$. Now take
$$
M_2=\begin{bmatrix}
1 & -y-b \\
0^T & I
\end{bmatrix}.
$$
Then we have
\begin{eqnarray*}
M_2(M_1AM_1^T)M_2^T &= \begin{bmatrix}
	\a & b \\
	b^T & D
\end{bmatrix}\\
&=\left[\begin{array}{@{}c|c|c@{}}
	\a& 0 & c\\\hline
	0^T & D_1& 0\\\hline
	c& 0& 0
	 \end{array}\right]
\end{eqnarray*}
for some nonzero $c$. Here $D_1$ is obtained by writing $D$ as a block diagonal matrix with a bottom diagonal zero.  Note that $A$ and $M_2(M_1AM_1^T)M_2^T$ correspond to the same standard form of quadratics, up to constant $c$ (Definition \ref{standard}). In fact, the quadratic form corresponding to $M_2(M_1AM_1^T)M_2^T$ can be obtained from the quadratic form corresponding to $A$ by  a linear change of variables. The quadratic form corresponding to $A$ and  $B$ have the same standard form also (up to $c$). That is: $A$ is hyperbolic if and only $B$ is hyperbolic. This is true since the difference between the two quadratic forms corresponding to $D_1$ and the full matrix is of the form
$(\alpha +2c)XY$ for two independent linear forms $X$ and $Y$ that are linearly independent of those linear forms appearing in $D_1$. Now let $B\in S(2t-2, m-1)$ be hyperbolic with $f_{k-1}^\delta(B)=\a$, then
$$
\phi_1^{-1}(B)=\left\{\begin{bmatrix}
-\a & u \\
u^T & B
\end{bmatrix}: u\notin\RS(B)\right\}.
$$
Consequently,
	$$
|\phi_1^{-1}(B)|= q^{m-1}-q^{2t-2}.
$$
\end{proof}


\begin{corollary}
	\label{cor: fiberIell}
	Let $B\in S(2t-2, m-1)$ and $A\in S(2t, m)$ be such that by deleting the first row and first column of $A$ we get $B$. Then $A$ is elliptic iff $B$ is elliptic. Further, for $B\in S(2t-2, m-1)$,
	
		$$
	|\phi_2^{-1}(B)|=\begin{cases*}
		q^{m-1}-q^{2t-2}, \text{ if B is elliptic,} \\
		0,\quad\quad\quad\quad\quad \text{ otherwise.}
	\end{cases*}
$$
\end{corollary}
In the next lemma, we understand the fibers $\phi_1^{-1}(B)$ for $B\in S(2t, m-1)$.
\begin{lemma}
	\label{lemma: fiberIIhyp}
	Let $B\in S(2t, m-1)$ and $A\in S(2t, m)$ be such that by deleting the first row and first column of $A$ we get $B$. Then $A$ is hyperbolic iff $B$ is hyperbolic. Further, If $B\in S(2t-2, m-1)$, then
	$$
|\phi_2^{-1}(B)|=\begin{cases*}
q^{2t-1}+ q^{t}- q^{t-1}, \text{ if B is hyperbolic and }f_{k-1}^\delta(B)=0, \\
q^{2t-1}- q^{t-1},\;\;\;\text{ if B is hyperbolic and }f_{k-1}^\delta(B)\neq0,\\
0, \quad\quad\quad\quad\quad\text{ otherwise.}
\end{cases*}
$$
\end{lemma}
\begin{proof}
Let $B\in S(2t, m-1)$ and let $L\in GL_{m-1}(\Fq)$ be a nonsingular matrix	such that the quadratic form corresponding to $LBL^T$ is in the standard hyperbolic or elliptic form. If $A\in S(2t, m)$ be matrix such that by deleting the first row and first column of $A$ we get $B$, then
$$
A=\begin{bmatrix}
uBu^T & uB\\
Bu^T& B
\end{bmatrix}.
$$
 If we take
 $$
 M=\begin{bmatrix}
 1 & -u\\
 0^T& L
 \end{bmatrix},
 $$
 then
$$
MAM^T=\begin{bmatrix}
0 & 0\\
0^T& LBL^T
\end{bmatrix}.
$$
This proves that $A$ and $B$ have the same standard form, i.e. $A$ is hyperbolic (elliptic) if and only if $B$ is hyperbolic (elliptic). This proves the first part of the lemma. It is now clear that if $B\in S(2t, m-1)$ is elliptic then $\phi_1^{-1}(B)=\emptyset$. On the other hand, if $B\in S(2t, m-1)$ is hyperbolic, we have two different cases. The first case, if $f_{k-1}^\delta(B)=0$, then
$$
\phi_1^{-1}(B)=\{A|A=\begin{bmatrix}
uBu^T & uB\\
Bu^T& B
\end{bmatrix}: uBu^T=0 \}
$$
Since $B$ is hyperbolic we get from equation \eqref{eq: noofhyp} that
$$
|\phi_1^{-1}(B)|= q^{2t-1}+ q^{t}- q^{t-1}.
$$
The second case is when $B\in S(2t, m-1)$ is hyperbolic and $f_{k-1}^\delta(B)=\a$ for some $\a\in\Fq^*$. In this case, we have
$$
\phi_1^{-1}(B)=\left\{A=\begin{bmatrix}
-\a & uB\\
Bu^T& B
\end{bmatrix}: uBu^T+ \a=0 \right\}.
$$
Using hyperbolic nature of $B$, parabolic nature of $B_\a$, and equation \eqref{eq: gammaalpha}, we get
$$
|\phi_1^{-1}(B)|= q^{2t-1}- q^{t-1}.
$$
This completes the proof of the lemma.
\end{proof}	
The fiber in the elliptic case is similar. We only need to use the solution for elliptic quadratic form from equation \eqref{eq: noofhyp}.

\begin{corollary}
	\label{cor: fiberIIellip}
	Let $B\in S(2t, m-1)$ and $A\in S(2t, m)$ be such that by deleting the first row and first column of $A$ we get $B$. Then $A$ is elliptic iff $B$ is elliptic. Further, If $B\in S(2t-2, m-1)$, then
	$$
	|\phi_2^{-1}(B)|=\begin{cases*}
	q^{2t-1}- q^{t}+ q^{t-1}, \text{ if B is elliptic and }f_{k-1}^\delta(B)=0, \\
	q^{2t-1}+ q^{t-1},\;\;\;\text{ if B is elliptic and }f_{k-1}^\delta(B)\neq0,\\
	0, \quad\quad\quad\quad\quad\text{ otherwise.}
	\end{cases*}
	$$
\end{corollary}
In the next lemma we count  the cardinalities of fibers $\phi_1^{-1}(B)$, and  $\phi_2^{-1}(B)$ for  $B\in S(2t-1, m-1)$ satisfying $f_{k-1}^\delta(B)=0$.

\begin{lemma}
	\label{lemma: fiberIII-1}
	Let $B\in S(2t-1, m-1)$ be a matrix satisfying $f_{k-1}^\delta(B)=0$.Then
	$$
	|\phi_1^{-1}(B)|=\frac{(q-1)}{2}q^{t-1}(q^{t-1}+1),
	$$
	and
		$$
	|\phi_2^{-1}(B)|=\frac{(q-1)}{2}q^{t-1}(q^{t-1}-1)
	$$
\end{lemma}
\begin{proof}
Let $B\in S(2t-1, m-1)$ with $f_{k-1}^\delta(B)=0$. There exists a nonsingular matrix  $L\in GL_{m-1}(\Fq)$ such that $LBL^T$ gives the quadratic form $cX_1^2+ X_2X_3+\cdots+ X_{2t-2}X_{2t-1}$ for some $c\in\Fq^*$. If $A\in S(2t, m)$ is the matrix such that by deleting the first row and columns of $A$ we obtain $B$ and $f_{k}^\delta(A)=0$, then
$$
A=\begin{bmatrix}
0 & uB\\
Bu^T& B
\end{bmatrix}\text{ such that }uBu^T\neq 0.
$$
Clearly,
$$
A\in \phi_1^{-1}(B) (\phi_2^{-1}(B))\iff A\text{ is hyperbolic }(elliptic).
$$
Now, if we take
$$
M= \begin{bmatrix}
1 & -u\\
u^T& L
\end{bmatrix},
$$	
then we have
$$
MAM^T= \begin{bmatrix}
-uBu^T & 0\\
0^T& LBL^T
\end{bmatrix}.
$$
Note that, $\rank(MAM^T)=\rank(A)=2t$, therefore, $uBu^T\neq 0$. If $uBu^T=\beta$, then the quadratic form corresponding to $MAM^T$ is
$$
-\beta X_0^2 + cX_1^2+ X_2X_3+\cdots+ X_{2t-2}X_{2t-1}.
$$
Therefore, the matrix $MAM^T$ and hence $A$ is hyperbolic (elliptic) if and only $-\frac{\beta}{c}$ is square(non-square). Now, let $\beta\in \Fq^*$ be fixed such that $-\frac{\beta}{c}$ is a square. Then every
$$
A=\begin{bmatrix}
0 & uB\\
Bu^T& B
\end{bmatrix}\text{ such that }uBu^T=-\beta
$$ is hyperbolic and hence lies in $\phi_1^{-1}(B)$. In this case the matrix $B_\beta$ is hyperbolic, therefore from equation \eqref{eq: gammaalpha}, we get the number of such $A$ is $q^{t-1}(q^{t-1}+1)$. Further, for a fixed $B\in S(2t-1, m-1)$ as in the beginning of the proof, there are $(q-1)/2$ $\beta\in\Fq^*$ such that $-\frac{-\beta}{c}$ is square. This gives
$$
|\phi_1^{-1}(B)| = \frac{(q-1)}{2}q^{t-1}(q^{t-1}+1).
$$
Repeating a similar argument in the case when $-\frac{\beta}{c}$ is non-square and keeping in mind that in this case, the matrix $B_\beta$ is elliptic, we get
$$
|\phi_2^{-1}(B)| = \frac{(q-1)}{2}q^{t-1}(q^{t-1}-1).
$$
\end{proof}	

In the previous cases, we have found the cardinalities of the fibers $\phi_i^{-1}(B)$ for $i=1, 2$ in all cases except in the case $B\in S(2t-1, m-1)$ satisfying $f_{k-1}^\delta(B)\neq 0$. In all the cases we have treated, the cardinalities of fibers $\phi_i^{-1}(B)$ for $i=1, 2$ are uniform. On the contrary, the fibers $\phi_i^{-1}(B)$ for $i=1, 2$ for $B\in S(2t-1, m-1)$ satisfying $f_{k-1}^\delta(B)\neq 0$ are irregular, i.e., they may have different cardinalities for different $B\in S(2t-1, m-1)$. It is true, however, that in the case  $B\in S(2t-1, m-1)$ satisfying $f_{k-1}^\delta(B)\neq 0$, the \underline{difference} of the cardinalities of fibers $\phi_1^{-1}(B)$ and $\phi_2^{-1}(B)$ is constant. To be more precise; we have the following lemma:


\begin{lemma} \label{difference}
	\label{lemma: fiberIII-2}
	Let $B\in S(2t-1, m-1)$ be a matrix satisfying $f_{k-1}^\delta(B)\neq0$. 
	Then
	$$
	|\phi_1^{-1}(B)|- |\phi_2^{-1}(B)|= -q^{t-1}.
	$$
\end{lemma}

\begin{proof}
	
	Let $B\in S(2t-1, m-1)$ be a matrix satisfying $f_{k-1}^\delta(B)=\a\in\Fq^*$. Let $L\in GL_{m-1}(\Fq)$ be a non-singular matrix such that the quadratic form corresponding to $LBL^T$ is $cX_1^2+X_2X_3+\cdots+X_{2t-2}X_{2t-1}$ for linearly independent linear forms $X_i,$, where $i=1,2,\cdots,2t-1$, and $c \ne 0$. Let $A\in S(2t, m)$ be such that $A\in \phi_i^{-1}(B)$, then
	$$
	A=\begin{bmatrix}
	-\a & uB\\
	Bu^T& B
	\end{bmatrix}\text{ such that }uBu^T\neq -\a.
	$$
	Clearly,
	$$
	A\in \phi_1^{-1}(B) (\phi_2^{-1}(B))\iff A\text{ is hyperbolic (elliptic). }
	$$
	Now, if we take
	$$
	M= \begin{bmatrix}
	1 & -u\\
	u^T& L
	\end{bmatrix},
	$$	
	then we have
	$$
A^\prime=	MAM^T= \begin{bmatrix}
	-\a-uBu^T & 0\\
	0^T& LBL^T
	\end{bmatrix}.
	$$
Note that both $A$ and $A^\prime$ are of the same type.	We have the following:
\begin{itemize}
	\item[1.] $\frac{\alpha+uBu^T}{c}\ne 0$.
	\item[2.] The quadratic form of $A$ is hyperbolic iff
	$\frac{\alpha+uBu^T}{c}$ is a non-zero square in $\mathbb{F}_q$.
	\item[3.] The quadratic form of $A$ is elliptic iff
	$\frac{-\alpha-uBu^T}{c}$ is a non-square in $\mathbb{F}_q$.
\end{itemize}
The first statement holds since the rank of $A$ is $2t$, not $2t-1$, which it would have been if the $\frac{\alpha+uBu^T}{c}$ were zero. The second and third statements follow since the quadratic form of $A'$ is of type:
$c(X_1^2-\frac{\alpha+uBu^T}{c}X_0^2)+X_2X_3+\cdots+X_{2t-2}X_{2t-1},$ and $(X_1^2-\frac{\alpha+uBu^T}{c}X_0^2)$ is factorizable into a product of two linear terms if and only $\frac{\alpha+uBu^T}{c}$ is a non-zero square.

So we set out to find the number of $uB$ such that
$\frac{\alpha+uBu^T}{c}$ is a non-zero square (thus giving hyperbolic $A$), and then the number of $uB$ such that
$\frac{\alpha+uBu^T}{c}$ is a non-square (thus giving elliptic $A$).

Given a fixed non-zero square $\gamma^2$ we look at the equation:
$$\frac{\alpha+uBu^T}{c}=\gamma^2.$$
We the substitute:
$v=uL^{-1}$, and the equation transforms to:
$$\frac{\alpha+v(LBL^T)v^T}{c}=\gamma^2.$$

The number of $v(LBL^T)=uBL^T$ satisfying this equation is then equal to the number of $uB$ satisfying this equation
(since $L$ is invertible), which again is equal to the number of $uB$ satisfying the original and identical  equation $$\frac{\alpha+uBu^T}{c}=\gamma^2.$$
We then prefer to  look at the number of $vB'$ satisfying $$\frac{\alpha+vB'v^T}{c}=\gamma^2,$$
for $B'=LBL^T$
since we already have introduced the quadratic form associated to $B'$ (and thereby the coefficient $c$).
We now insert $vB'v^T=c(X_1^2+X_2X_3+\cdots+X_{2t-2}X_{2t-1})$ again and obtain that the previous equation can be expressed as
$$X_1^2+X_2'X_3+\cdots+X_{2t-2}'X_{2t-1}-(\gamma^2-\frac{\alpha}{c})=0,$$
for linearly independent $X_i$ and $X_j'$ (here we have set $X_j'=X_j/c$ for even $j$).
For such a fixed $\gamma$ the number of solutions is (what we will call) $v_H,v_P,v_E$, according to whether
$\gamma^2-\frac{\alpha}{c}$ is a non-zero square, zero, or a non-square respectively. Here:
	
	\begin{itemize}
		\item[(i)] $v_H$ is the number of points on a projective quadric defined by a hyperbolic quadric of rank $2t$ minus the number of points on a projective quadric defined by a parabolic quadric of rank $2t-1.$
		\item[(ii)] $v_P$ is the number of points on an affine quadric defined by a parabolic quadric of rank $2t-1.$
		\item[(iii)] $v_E$ is the number of points on a projective quadric defined by an elliptic quadric of rank $2t$ minus the number of points on a projective quadric defined by a parabolic quadric of rank $2t-1.$
	\end{itemize}
Formulas for $v_H, v_P$, and $v_E$ can be found from equation \eqref{eq: noofsoln}. It is not hard to see that
\begin{equation}
\label{eq: hypminusell}
v_H-v_P=v_P-v_E=q^{t-1}.
\end{equation}

	Let $e$ be a fixed non-square in $F_q$. We will investigate the number of ``hyperbolic" $A'$; these are the ones corresponding to solutions to
	\begin{equation} \label{*}
	X_1^2+X_2'X_3+\cdots+X_{2t-2}'X_{2t-1}-(\gamma^2-\frac{\alpha}{c})=0,
	\end{equation}
for the $(q-1)/2$ non-zero squares $\gamma^2$, and thereafter the number of ``elliptic" $A'$; these are the ones corresponding to solutions to
	\begin{equation} \label{**}
	X_1^2+X_2'X_3+\cdots+X_{2t-2}'X_{2t-1}-(e\gamma^2-\frac{\alpha}{c})=0,
	\end{equation}
	 for the $(q-1)/2$ non-squares $e\gamma^2.$ We split into 4 cases:
	
	(1): Both $\frac{\alpha}{c}$ and  $-\frac{\alpha}{c}$ are non-squares.
	We investigate the condition  (\ref{**})
	$$\gamma^2-\frac{\alpha}{c}=T^2,$$ for $T \in \mathbb{F}_q$. This single equation is equivalent to:
	\begin{itemize}
		\item $\gamma-T=\frac{\alpha}{c}\beta$
		\item $\gamma+T=\frac{\alpha}{c}\beta^{-1},$
		simultaneously, for \underline{some} $\beta \in \mathbb{F}_q^*$.
	\end{itemize}
	The two last equations, taken alone, give $q-1$ solution pairs $(\gamma_0,T_0)$, one for each $\beta$ in question.
	Since neither $\alpha / c$ nor $-\alpha / c$ are squares,
	both $\gamma_0$ and $T_0$ are non-zero in all these solutions, so none of them correspond to the "forbidden" $\gamma=0$ that would have made $A'$ parabolic of rank $2t-1$, and none of them gives a zero for $(\gamma^2-\frac{\alpha}{c})$. The last observation gives that the number of solutions to  (\ref{*}) never is $v_P$, so for the $\gamma^2$ that matches some $T^2$ we really get $v_H$ solutions to  (\ref{*}). On the other hand the solutions to (\ref{**}) come $4$ by $4$ as
	$(\gamma_0,T_0)(\gamma_0,-T_0)(-\gamma_0,T_0)(-\gamma_0,-T_0)$, for each solution $(\gamma_0,T_0)$, and hence we get exactly $\frac{q-1}{4}$ values $\gamma^2$ that match some (in this case non-zero) $T^2$. Hence we get $v_H$ solutions to  (\ref{*}) for exactly $\frac{q-1}{4}$ values $\gamma^2$. Since no $\gamma^2$ matches a zero-valued $T$ all  the remaining $\frac{q-1}{4}$ values of $\gamma^2$ give that $\gamma^2-\frac{\alpha}{c}$ is a non-square, and we get $v_E$ solutions of  (\ref{*}) for these values.
	
	All in all, we get
	$$\frac{q-1}{4}v_H+\frac{q-1}{4}v_E$$ hyperbolic $A'$ in the fibre over $B'$.
	
	Now we will find the number of  elliptic $A'$ in the fiber over $B'$, corresponding to
	$$\frac{\alpha+v(LBL^T)v^T}{c}=N=e\gamma^2,$$
	for some non-square $N$. Here we think of $e$ as a fixed non-square, so we will find the non-zero $\gamma$ that satisfies this. We maintain condition (1) that both $\alpha / c$ and  $\alpha / c$ are non-squares.
	Now we look at the condition:
	\begin{equation} \label{***}
	N=e\gamma^2-\frac{\alpha}{c}=eT^2,
	\end{equation}
	 for $T \in \mathbb{F}_q$.
	If given $\gamma$, this condition holds for some non-zero $T$, we get $v_E$ solutions to  (\ref{*}), and if it holds for $T=0$, we get $v_P$ solutions. Otherwise, we get $v_H$ solutions. This single equation is equivalent to:
	$$\gamma^2-T^2=\frac{\alpha}{ce},$$ and furthermore:
	\begin{itemize}
		\item $\gamma-T=\frac{\alpha}{ce}\beta$
		\item $\gamma+T=\frac{\alpha}{ce}\beta^{-1},$
		simultaneously, for \underline{some} $\beta \in \mathbb{F}_q^*$.
	\end{itemize}
	The two last equations, again give $q-1$ solution pairs $(\gamma_0,T_0)$, one for each $\beta$ in question.
	Since both $\alpha / ec$ and $-\alpha / ec$ are squares,
	both $(0,T_0),(0,-T_0),(\gamma_0,0),(-\gamma_0,0)$ for some non-zero $\gamma_0$ and $T_0$ are among these $q-1$ solution pairs ($(0,0)$ can't be a solution).
	The pairs $(0,T_0),(0,-T_0)$ can be disregarded, since $\gamma=0$ is forbidden, since $A'$ has rank $2t$. The cases $(\gamma_0,0),(-\gamma_0,0)$, however, correspond to one value of $N=e\gamma^2$ with $v_P$ solutions.
	
	The remaining $q-5$ solutions to   (\ref{**}) come $4$ by $4$ as
	$$(\gamma_0,T_0)(\gamma_0,-T_0)(-\gamma_0,T_0)(-\gamma_0,-T_0),$$ for each solution $(\gamma_0,T_0)$, and hence we get exactly $\frac{q-5}{4}$ values $\gamma^2$ that match some (in this case non-zero) $T^2$. Hence we get $v_E$ solutions to  (\ref{*}) for exactly $\frac{q-5}{4}$ values $N=e\gamma^2$.  The remaining $\frac{q-1}{2}-1 - \frac{q-5}{4}=\frac{q-1}{4}$ values of $N=e\gamma^2$ then give $v_H$ elliptic $A'$ in the fibre over $B'$.
	
	All in all, we get
	$$\frac{q-1}{4}v_H+v_P+\frac{q-5}{4}v_E$$ elliptic $A'$ in the fibre over $B'$. Comparing with  the number of hyperbolic $A'$ in the fibre over $B'$, computed above, we conclude that the number of hyperbolic $A'$ in the fibre over $B'$ minus the number of elliptic $A'$ in the fibre over $B'$ is:
	$$(\frac{q-1}{4}v_H+\frac{q-1}{4}v_E)-(\frac{q-1}{4}v_H+v_P+\frac{q-5}{4}v_E)=$$
	\[v_E-v_P=-q^{t-1}.\]
	
	Case (2): $\frac{\alpha}{c}$ is a (non-zero) square and  $-\frac{\alpha}{c}$ is a  non-square. Using calculations as in Case (1) we obtain:
	
	$$\frac{q-3}{4}v_H+v_P+\frac{q-3}{4}v_E$$ hyperbolic $A'$ in the fibre over $B'$, and
	
	$$\frac{q+1}{4}v_H+\frac{q-3}{4}v_E$$ elliptic $A'$ in the fibre over $B'$. The difference is  $v_P-v_H=-q^{t-1}$.
	
	Case (3): $\frac{\alpha}{c}$ is a non-square and  $-\frac{\alpha}{c}$ is a ( non-zero square). Using calculations as in Case (1) we obtain:
	
	$$\frac{q-3}{4}v_H+\frac{q+1}{4}v_E$$ hyperbolic $A'$ in the fibre over $B'$, and
	
	$$\frac{q-3}{4}v_H+v_P+\frac{q-3}{4}v_E$$ elliptic $A'$ in the fibre over $B'$. The difference is  $v_E-v_P=-q^{t-1}$.
	
	Case (4): Both $\frac{\alpha}{c}$  and  $-\frac{\alpha}{c}$ are (non-zero) squares. Using calculations as in Case (1) we obtain:
	
	$$\frac{q-5}{4}v_H+v_P+\frac{q-1}{4}v_E$$ hyperbolic $A'$ in the fibre over $B'$, and
	
	$$\frac{q-1}{4}v_H+\frac{q-1}{4}v_E$$ elliptic $A'$ in the fibre over $B'$. The difference is  $v_P-v_H=-q^{t-1}$.
	
\end{proof}

Now, let $\theta_1$ be the average of $|\phi_1^{-1}(B)|$ for all $B\in S(2t-1, m-1)$ satisfying $f_{k-1}^\delta(B)\neq 0$. and let $\theta_2$ be the average of $|\phi_2^{-1}(B)|$ for all $B\in S(2t-1, m-1)$ satisfying $f_{k-1}^\delta(B)\neq 0$. The next corollary is an immediate consequence of Lemma \ref{difference}.
\begin{corollary}
	\label{cor: fiberIII-2}
	Let $\theta_1$ and $\theta_2$ be as above. Then
	$$
	\theta_1- \theta_2=-q^{t-1}.
	$$
\end{corollary}

Now we are ready to give an upper bound for  $h_{k}^\delta(2t, m)- e_{k}^\delta(2t, m)$.
\begin{proposition}
	\label{prop: bound}
	Let $2t\leq m$ and let $1\leq k\le m$.  Then
	\begin{eqnarray*}
		h_k^\delta(2t, m)- e_k^\delta(2t, m)&\leq \; q\;\dfrac{\prod\limits_{i=0}^{2t-2}(q^{m-1}-q^i)}{\prod\limits_{i=0}^{t-2}(q^{2t-2}-q ^{2i})} + q^{2t}\dfrac{\prod\limits_{i=0}^{2t-1}(q^{m-1}-q^i)}{\prod\limits_{i=0}^{t-1}(q^{2t}-q ^{2i})}
	\end{eqnarray*}
\end{proposition}

\begin{proof}
We have done most of the hard work. From equation \eqref{eq: restmapchiI} we get
	
\begin{align*}
	h_k^\delta(2t, m)&=  \sum_{\mathclap{\substack{B\in S(2t-2, m-1) }}}|\phi_1^{-1}(B)| \;\;\;+\;\;\; \sum_{\mathclap{\substack{B\in S(2t, m-1) \\ f_{k-1}^\delta(B)=0}}}|\phi_1^{-1}(B)|   \;\;\;+\;\;\; \sum_{\mathclap{\substack{B\in S(2t, m-1) \\ f_{k-1}^\delta(B)\neq 0}}}|\phi_1^{-1}(B)|\\
	& + \sum_{\mathclap{\substack{B\in S(2t-1, m-1) \\ f_{k-1}^\delta(B)= 0}}}|\phi_1^{-1}(B)|\;\;\; + \;\;\;\sum_{\mathclap{\substack{B\in S(2t-1, m-1) \\ f_{k-1}^\delta(B)\neq  0}}}|\phi_1^{-1}(B)|. \\
\end{align*}
From Lemmas \ref{lemma: fiberIhyp}, \ref{lemma: fiberIIhyp}, \ref{lemma: fiberIII-1}, and  the definition of $\theta_1$ above, we get:
$$
\begin{array}{lll}
 h_k^\delta(2t, m)&\leq &  v_{+1}(2t-2, m-1)(q^{m-1}- q^{2t-2}) + h_{k-1}^\delta(2t, m-1) (q^{2t-1}+ q^t -q^{t-1}) \\
& +& (v_{+1}(2t, m-1)- h_{k-1}^\delta(2t, m-1) ) (q^{2t-1}- q^{t-1}) \\
 &+ & p_{k-1}^\delta(2t-1, m-1)\frac{(q-1)}{2} q^{t-1}(q^{t-1}+1) \\  \vspace*{.2cm}
 &+& (\mu_{m-1}(2t-1)- p_{k-1}^\delta(2t-1, m-1))\theta_1\\
 &=& v_{+1}(2t-2, m-1)(q^{m-1}- q^{2t-2}) + h_{k-1}^\delta(2t, m-1) q^t\\
 &+ & v_{+1}(2t, m-1)(q^{2t-1}-q^{t-1})
+ p_{k-1}^\delta(2t-1, m-1)\frac{(q-1)}{2} q^{t-1}(q^{t-1}+1)\\
&+& (\mu_{m-1}(2t-1)- p_{k-1}^\delta(2t-1, m-1))\theta_1.
\end{array}
$$

\noindent Similarly, using equation \eqref{eq: restmapchiII}, Corollaries \ref{cor: fiberIell}, \ref{cor: fiberIIellip}, Lemma \ref{lemma: fiberIII-1} and the definition of  $\theta_2$ above,  we get:
$$
\begin{array}{lll}
e_k^\delta(2t, m)&\geq & v_{-1}(2t-2, m-1)(q^{m-1}- q^{2t-2}) - e_{k-1}^\delta(2t, m-1) q^t \\
&+& v_{-1}(2t, m-1)(q^{2t-1} + q^{t-1}) +p_{k-1}^\delta(2t-1, m-1)\frac{(q-1)}{2} q^{t-1}(q^{t-1}-1)\\
&+& (\mu_{m-1}(2t-1)- p_{k-1}^\delta(2t-1, m-1))\theta_2.
\end{array}
$$
Combining the above two inequalities, we get
$$
\begin{array}{lll}
 h_k^\delta(2t, m)- e_k^\delta(2t, m)&\leq& (v_{+1}(2t-2, m-1) - v_{-1}(2t-2, m-1))(q^{m-1}- q^{2t-2})\\
& + & (h_{k-1}^\delta(2t, m-1) + e_{k-1}^\delta(2t, m-1)) q^t + (v_{+1}(2t, m-1) \\
&-& v_{-1}(2t, m-1))q^{2t-1} - (v_{+1}(2t, m-1) + v_{-1}(2t, m-1)) q^{t-1} \\
&+& p_{k-1}^\delta(2t-1, m-1)(q-1)q^{t-1}\\
&+&(\mu_{m-1}(2t-1)- p_{k-1}^\delta(2t-1, m-1))(\theta_1-\theta_2).
\end{array}
$$
Substituting the trivial bounds $h_{k-1}^\delta(2t, m-1) + e_{k-1}^\delta(2t, m-1)\leq \mu_{m-1}(2t))$, $p_{k-1}^\delta(2t-1, m-1)\leq \mu_{m-1}(2t-1)$ and using Corollary \ref{cor: fiberIII-2} saying $\theta_1-\theta_2=-q^{t-1}$, we get
$$
\begin{array}{lll}
 h_k^\delta(2t, m)- e_k^\delta(2t, m)&\leq& (v_{+1}(2t-2, m-1) - v_{-1}(2t-2, m-1))(q^{m-1}- q^{2t-2})\\
& +& \mu_{m-1}(2t) q^t + (v_{+1}(2t, m-1) - v_{-1}(2t, m-1))q^{2t}\\
&  - & \mu_{m-1}(2t-1) q^{t-1} + \mu_{m-1}(2t-1)q^{t} - {\mu_{m-1}(2t-1)})q^{t-1}.\\
\end{array}
$$
Here we also used the fact that $v_{+1}(2t, m-1) + v_{-1}(2t, m-1)=\mu_{m-1}(2t)$. Now from \cite[Prop. 2.4]{Schmidt}, and equation \eqref{eq: noofhyp}, we have
\begin{equation}
\label{eq: diffhypell}
v_{+1}(2t, m) - v_{-1}(2t, m)=\dfrac{\prod\limits_{i=0}^{2t-1}(q^m-q^i)}{\prod\limits_{i=0}^{t-1}(q^{2t}-q ^{2i})},
\end{equation}
$$
\mu_{m}(2t)=v_{+1}(2t, m) + v_{-1}(2t, m)=q^t\dfrac{\prod\limits_{i=0}^{2t-1}(q^m-q^i)}{\prod\limits_{i=0}^{t-1}(q^{2t}-q ^{2i})},
$$
and
$$
\mu_{m}(2t-1)=\frac{1}{q^{t-1}}\dfrac{\prod\limits_{i=0}^{2t-2}(q^{m}-q^i)}{\prod\limits_{i=0}^{t-2}(q^{2t-2}-q ^{2i})}.
$$
Substituting these values in the inequality above, we get
$$
\begin{array}{lll}
 h_k^\delta(2t, m)- e_k^\delta(2t, m)&\leq & \dfrac{\prod\limits_{i=0}^{2t-3}(q^{m-1}-q^i)}{\prod\limits_{i=0}^{t-2}(q^{2t-2}-q ^{2i})}(q^{m-1}- q^{2t-2})\;
  +\; q^t\cdot q^t\dfrac{\prod\limits_{i=0}^{2t-1}(q^{m-1}-q^i)}{\prod\limits_{i=0}^{t-1}(q^{2t}-q ^{2i})} \\
 &+ & q^{2t-1}\dfrac{\prod\limits_{i=0}^{2t-1}(q^{m-1}-q^i)}{\prod\limits_{i=0}^{t-1}(q^{2t}-q ^{2i})}
   \;-\; q^{t-1}\cdot q^t\dfrac{\prod\limits_{i=0}^{2t-1}(q^{m-1}-q^i)}{\prod\limits_{i=0}^{t-1}(q^{2t}-q ^{2i})}\\
   & +& (q-1)q^{t-1}\frac{1}{q^{t-1}}\dfrac{\prod\limits_{i=0}^{2t-2}(q^{m-1}-q^i)}{\prod\limits_{i=0}^{t-2}(q^{2t-2}-q ^{2i})}\\
   &=& \dfrac{\prod\limits_{i=0}^{2t-2}(q^{m-1}-q^i)}{\prod\limits_{i=0}^{t-2}(q^{2t-2}-q ^{2i})}
   \;+\; q^{2t}\dfrac{\prod\limits_{i=0}^{2t-1}(q^{m-1}-q^i)}{\prod\limits_{i=0}^{t-1}(q^{2t}-q ^{2i})}\\
  & + & (q-1)\dfrac{\prod\limits_{i=0}^{2t-2}(q^{m-1}-q^i)}{\prod\limits_{i=0}^{t-2}(q^{2t-2}-q ^{2i})}\\
   &=& q\;\dfrac{\prod\limits_{i=0}^{2t-2}(q^{m-1}-q^i)}{\prod\limits_{i=0}^{t-2}(q^{2t-2}-q ^{2i})} \;+\; q^{2t}\dfrac{\prod\limits_{i=0}^{2t-1}(q^{m-1}-q^i)}{\prod\limits_{i=0}^{t-1}(q^{2t}-q ^{2i})}.
\end{array}
$$

 This completes the proof of the proposition.
\end{proof}	
 Now we are ready to prove the main theorem of the article, i.e., we are ready to compute the minimum distance of the code $\CAt$.

\begin{theorem}
	\label{thm: mainthm}
	Let $2t\leq m$ be positive integers. The minimum distance of the code $\CAt$ is $W_1^\delta(2t, m)$, where $W_1^\delta(2t, m)$ is given by
 $$
(q-1)q^{m-1}\nu_{m-1}(2t-2) + (q-1)q^{2t-1}\mu_{m-1}(2t-1)
	 + (q-1)(q^{2t-1}-q^{-1})\mu_{m-1}(2t)
 $$
\end{theorem}
\begin{proof}
	In the view of Corollary \ref{cor: wtdiff},  it is enough to prove that for $1\leq k\le m$ and $2t\le m$,
	$$
	v_{+1}(2t, m) - v_{-1}(2t, m) \geq h_k^\delta(2t, m) - e_k^\delta(2t, m).
	$$
	From equation \eqref{eq: diffhypell} and Proposition \ref{prop: bound}, we get
 $$
\begin{array}{lll}
  &(v_{+1}(2t, m) - v_{-1}(2t, m) ) \;-\; (h_k^\delta(2t, m) - e_k^\delta(2t, m)) \vspace*{.2cm}\\
&\geq \dfrac{\prod\limits_{i=0}^{2t-1}(q^m-q^i)}{\prod\limits_{i=0}^{t-1}(q^{2t}-q ^{2i})} \;-\; q^{2t}\dfrac{\prod\limits_{i=0}^{2t-1}(q^{m-1}-q^i)}{\prod\limits_{i=0}^{t-1}(q^{2t}-q ^{2i})}
\;-\; q\;\dfrac{\prod\limits_{i=0}^{2t-2}(q^{m-1}-q^i)}{\prod\limits_{i=0}^{t-2}(q^{2t-2}-q ^{2i})} \\
&= \dfrac{\prod\limits_{i=0}^{2t-2}(q^{m-1}-q^i)}{\prod\limits_{i=0}^{t-1}(q^{2t}-q ^{2i})} \left( q^{2t-1}(q^m-1) - q^{2t}(q^{m-1}- q^{2t-1})\right) - q\;\dfrac{\prod\limits_{i=0}^{2t-2}(q^{m-1}-q^i)}{\prod\limits_{i=0}^{t-2}(q^{2t-2}-q ^{2i})}\\
&= \dfrac{\prod\limits_{i=0}^{2t-2}(q^{m-1}-q^i)}{\prod\limits_{i=0}^{t-1}(q^{2t}-q ^{2i})} \left( q^{2t-1}(q^{2t}-1)\right) - q\;\dfrac{\prod\limits_{i=0}^{2t-2}(q^{m-1}-q^i)}{\prod\limits_{i=0}^{t-2}(q^{2t-2}-q ^{2i})}\\
&=\dfrac{\prod\limits_{i=0}^{2t-2}(q^{m-1}-q^i)}{\prod\limits_{i=0}^{t-2}(q^{2t-2}-q ^{2i})} \left( \frac{q^{2t-1}(q^{2t}-1)}{q^{2t-2}(q^{2t}-1)}- q\right)\\
&=0.
\end{array}
 $$


\end{proof}

\begin{corollary}
	\label{cor: minimumdistance}
	 The minimum distance of the code $\widehat{C}_{symm}(2t, m)$ is $d$, where
	 $$
	 d=q^{m-1}\nu_{m-1}(2t-2) + q^{2t-1}\mu_{m-1}(2t-1)
	 + (q^{2t-1}-q^{-1})\mu_{m-1}(2t).
	 $$
\end{corollary}
\begin{proof}
	The proof of the corollary follows from Theorem \ref{thm: WeightW1}, Theorem \ref{thm: mainthm}, and equation \eqref{eq: wtafftoDet}.
	
\end{proof}

 Let $m$ be a positive integer and $1\le 2t+1< m$.	For the symmetric determinantal code $C_{symm}(2t+1, m)$, we propose the following conjecture:

\begin{conjecture}
	The minimum distance of the code ${C}_{symm}(2t+1, m)$ is given by $W_2^\delta(2t+1, m)$  where $-\delta\in\Fq^*$ is a square. Furthermore, there exists a positive $\mathcal{E}(2t+1, m)$ such that
		$$
		W_2^{\delta_1}(2t+1, m)  	< 	W_1^{\delta}(2t+1, m) <	W_2^{\delta_2}(2t+1, m)
		$$
		where $-\delta_1$ is a square and $-\delta_2$ is a non square, and:
		$$
		W_2^{\delta_1}(2t+1, m)=	W_1^{\delta}(2t+1, m) \pm \mathcal{E}(2t+1, m).
  $$
	
\end{conjecture}
 Using SAGE, we computed the weights $\Wktd$ for some small values of $m$, and for a fixed $m$, we looked at all possible values of $k$ and $t$. Studying those results, we here give formulas for the weights $\Wktd$ over $\Fq$. In the  tables below,  we list them all. The table is over $\Fq$ and $\delta\in\Fq^*$ is either a square or a nonsquare.
\\
\\
\begin{tabular}{ |p{3cm}|p{3cm}|p{3cm}|p{3cm}|  }
	\hline
	\bf m=3 & \bf t=1 & \bf t=2 & \bf t=3  \\
	\hline
$ W_1^{\delta}(t, m) $   & $q^2(q-1)$ & $q^4(q-1)$ &$q^5(q-1)$\\
	\hline
$ W_2^\delta(t, m) $  &  $q^2(q-1)\pm q(q-1)$ & $q^4(q-1)$ &  $q^5(q-1)$\\
\hline
$W_3^\delta(t, m)$&   $q^2(q-1)$ &  $q^4(q-1)+ q^2(q-1) $   &$q^5(q-1)$\\
	\hline
\end{tabular}
\\
\\
\\
\begin{tabular}{ |p{2cm}|p{2.5cm}|p{2.5cm}|p{2.5cm}|p{2.5cm}|  }
	\hline
	\bf m=4 & \bf t=1 & \bf t=2 & \bf t=3 & \bf t=4 \\
	\hline
	$ W_1^{\delta}(t, m) $   & $q^3(q-1)$ & $q^6(q-1)$ &$q^8(q-1)+ q^5(q-1)^2$ & $q^9(q-1)$\\
	\hline
	$ W_2^\delta(t, m) $  &  $q^3(q-1)\pm q^2(q-1)$ & $q^6(q-1)$ & $W_1^\delta(3, 4)\pm q^4(q-1)$ & $q^9(q-1)$\\
	\hline
	$W_3^\delta(t, m)$&   $q^3(q-1)$ &  $q^6(q-1)+ q^4(q-1) $   &$W_1^\delta(3, 4)$ & $q^9(q-1)$\\
	\hline
		$W_4^\delta(t, m)$&   $q^3(q-1)\pm q(q-1)$ &  $q^6(q-1)+ q^4(q-1) $   & $W_1^\delta(3, 4)\pm q^3(q-1)$ & $q^9(q-1)$\\
			\hline
\end{tabular}
\\
\\
\\
\begin{tabular}{ |p{1.5cm}|p{2cm}|p{2cm}|p{3cm}|p{2cm}| p{2cm}| }
	\hline
	\bf m=5 & \bf t=1 & \bf t=2 & \bf t=3 & \bf t=4 & \bf t=5\\
	\hline
	$ W_1^{\delta}(t, m) $   & $q^4(q-1)$ & $q^8(q-1)$ &$q^{11}(q-1)+ q^8(q-1)^2 + q^6(q-1)(q^2-1)$ & $q^{13}(q-1) + q^{10}(q-1)^2$  & $q^{14}(q-1)$\\
	\hline
	$ W_2^\delta(t, m) $  &  $q^4(q-1)\pm q^3(q-1)$ & $W_1^\delta(2, 5)$ & $W_1^\delta(3, 5)\pm (q^7(q-1)^2 + q^5(q-1)(q^2-1))$ & $W_1^\delta(4, 5)$ &  $q^{14}(q-1)$\\
	\hline
	$W_3^\delta(t, m)$&   $q^4(q-1)$ &  $W_1^\delta(2, 5) + q^6(q-1)$   &$W_1^\delta(3, 5)$ & $W_1^\delta(4, 5) + q^8(q-1)^2$  &  $q^{14}(q-1)$\\
	\hline
	$W_4^\delta(t, m)$&  $q^4(q-1)\pm q^2(q-1)$ &  $W_3^\delta(2, 5)$   &  $W_1^\delta(3, 5)\pm (q^6(q-1) + q^4(q-1)(q^2-1))$ & $W_3^\delta(4, 5)$  &  $q^{14}(q-1)$\\
	\hline
	$W_5^\delta(t, m)$&   $q^4(q-1)$ &  $W_3^\delta(2, 5) + q^4(q-1)$   & $W_1^\delta(3, 5)$ &  $W_3^\delta(4, 5) -q^6(q-1)$ &  $q^{14}(q-1)$\\
	\hline
\end{tabular}

\begin{remark}
The observation that  $W^{\delta}_2(t,m)=W^{\delta}_1(t,m)$, for all even $t$, and all $\delta$, for $m=3,4,5$ can be generalized to the same result for all natural numbers $m \ge 2$ after a refined study of the proof of Proposition \ref{prop: bound}. So for even $t$ the minimum distance $d$ is computed by the $W^{\delta}_2(t,m)$ also.
The table clearly indicates that for even $t$ the codewords of  $\C$ have exactly $[\frac{m+1}{2}] -1$ different weights.
\end{remark}

\section{Acknowledgements}
Peter Beelen would like to acknowledge the support from The Danish Council for Independent Research (DFF-FNU) for the project \emph{Correcting on a Curve}, Grant No.~8021-00030B. Trygve Johnsen is supported by grant 280731 from the Research Council of Norway (RCN). This work was also partially supported by the project Pure Mathematics in Norway, funded by Bergen Research Foundation and Tromsø Research Foundation. Prasant Singh is with the Department of Mathematics, IIT Jammu. He would like to thank the RCN Grant- 280731, and the Department of Mathematics and Statistics, UiT as  a major part of the work in this article was done when he was a postdoc at UiT.

\end{document}